\documentclass[12pt]{article}
\usepackage{amssymb}
\usepackage{amsbsy}
\usepackage{amsthm}
\usepackage{amsfonts}
\usepackage{amsmath}
\usepackage{setspace}
\usepackage{parskip}
\usepackage{geometry}
\usepackage{scalefnt}
\usepackage{hyperref}
\usepackage{cleveref}
\usepackage{enumerate}
\usepackage{natbib}
\usepackage{bm}
\usepackage[bottom]{footmisc}
\usepackage{verbatim}
\usepackage{booktabs}
\usepackage{multirow}
\usepackage{graphicx}
\usepackage{float}
\usepackage{caption}
\usepackage{xcolor}
\usepackage{placeins}

\newtheorem{theorem}{Theorem}

\newtheorem{assumption}{Assumption}[section]

\newtheorem{lemma}{Lemma}[section]

\newtheorem{proposition}{Proposition}

\newtheorem{remark}{Remark}

\renewenvironment{proof}[1][Proof]{\noindent\textbf{#1.} }{\ \rule{0.5em}{0.5em}}

\newcommand{\wh}{\widehat}
\newcommand{\wt}{\widetilde}

\title{\Large Estimation of Random-Coefficient Dynamic Panel Data Models \\ with a Fixed $T$\thanks{We thank Xun Lu, Philip Marx, Matthew Masten and Elie Tamer for helpful discussion and comments.}}
\author{Xun Tang\thanks{Department of Economics, Rice University. Email: xun.tang@rice.edu} \and Pei Yu\thanks{Department of Economics, Rice University. Email: pei.yu@rice.edu}}
\date{\today}

\begin{document}

{\hypersetup{pdfborder={0 0 0}}\maketitle}

\vspace{1.5em}

\begin{abstract}

We study dynamic linear panel data models in which the lagged outcomes and strictly exogenous covariates carry individual-specific coefficients and the time-varying errors have a flexible covariance structure.
With a fixed number of time periods, we point-identify the joint distribution of the random coefficients and the structural errors under a distributional form of strict exogeneity, and propose a closed-form, multi-step estimator based on the inverse Radon transform.
We establish a uniform convergence rate for the estimator of the random coefficient density, as well as uniform consistency of the estimator for the conditional density of the time-varying structural errors.
Monte Carlo simulations demonstrate good finite-sample performance of the estimators.

\end{abstract}

\noindent\textbf{Keywords:} dynamic panel data; random coefficients; heterogeneous state dependence; fixed-$T$; inverse Radon transform

\newpage

\section{Introduction}\label{sec:intro}

Dynamic panel data models capture the persistence of economic outcomes, such as employment, earnings, or firm performance, by including lagged dependent variables among the covariates in structural equations.
They are instrumental for disentangling ``genuine'' state dependence, or the direct impact of past outcomes on current ones, from ``spurious'' dependence caused by persistent unobserved heterogeneity across individuals.
These distinct sources of dependence carry sharply different implications, because whether a transitory shock or a temporary policy intervention could leave a lasting imprint depends on the strength of true state dependence.

As such, dynamic panel models have become a workhorse in fields such as labor economics and industrial organization. Canonical applications include dynamic employment in \citet{ArellanoBond1991, BlundellBond1998}, production functions in \citet{blundell2000gmm}, corporate performance and competition in \citet{nickell1996competition}, firm innovation in \citet{aghion2005competition}, market shares and stock valuation in \citet{blundell1999market}, investments and financial factors in \citet{bond1994dynamic,bond2003financial}.

Earlier econometrics literature had studied dynamic panel models with constant autoregressive coefficients where the state dependence is homogeneous across units.
More recent works accommodated heterogeneous state dependence in workers, households, or firms by including random coefficients on the lagged dependent variables, and found it empirically salient in many settings, e.g., labor income dynamics in \citet{FernandezValGaoLiaoVella2022} in a large-$N$, large-$T$ setting.\footnote{
    Relatedly, \citet{LiuMoonSchorfheide2020} and \citet{LuMiaoSu2024} applied dynamic panel models with random coefficients on strictly exogenous covariates in the contexts of forecasting and program evaluation.}

Bringing random coefficients to bear on dynamic panels raises two challenges that motivate this paper.
The first is to let the lagged dependent variable carry a random coefficient in a short panel.
Because the lagged outcome is predetermined rather than strictly exogenous, heterogeneous autoregressive dynamics are difficult to identify when the number of time periods $T$ is small.
Nonetheless, knowledge of the distribution of the autoregressive coefficient, rather than a single common value or its mean in the population, is needed for answering many questions.
Examples include measuring the mass of units with non- or near-unit-root persistence (which governs long-run responses), and quantifying the spread and skewness of persistence (which affects the shape of the distribution of forecasts or policy targets).
Whether large-$T$ or fixed-$T$ asymptotics is the more suitable framework depends on the application and the data at hand, and neither approach dominates the other; our aim is to complement this developing literature by offering a methodological alternative that is valid when $T$ is fixed.

The second challenge is to recover the joint distribution of the random coefficients \emph{together with} the time-varying errors (random intercepts in the \emph{structural} form), rather than selected moments or a marginal distribution. The distribution of random coefficients reveals how state dependence correlates with heterogeneous responses to the covariates, both of which are persistent; that of the random structural intercepts captures the distribution of time-varying shocks (permanent and transitory) underlying the process.
Knowledge of their joint distribution is needed not only to separate and quantify the stochastic contribution by permanent vs transitory components, but also to simulate counterfactual trajectories in outcomes such as earnings or productivity. Means or low-order moments from their marginal distributions alone cannot answer these questions.

We study the dynamic linear panel data model
\begin{equation*}
    Y_{it} = \gamma_i Y_{it-1} + X_{it}'\beta_{it} + W_{it}'\delta_t + U_{it}, \qquad t = 1,\dots,T,
\end{equation*}
where the autoregressive coefficient $\gamma_i$ and the coefficients $\beta_{it}$ on the strictly exogenous $X_{it}$ are individual-specific, while the coefficients $\delta_t$ on the sequentially exogenous $W_{it}$ are common across units $i$.
The time-varying errors, or ``random structural intercepts'', $U_{it}$, absorb both fixed effects and transitory shocks (e.g., $U_{it}=\alpha_i+\varepsilon_{it}$).
Our object of interest is the joint distribution of $(\gamma_i,\beta_i,U_i)$ conditional on the initial $Y_{i0}$, along with $\delta_t$.

The model identification requires a \emph{distributional} form of strict exogeneity: $(\gamma_i,\beta_i,U_i)$ are jointly independent of the exogenous covariate history conditional on $Y_{i0}$ (Assumption~\ref{assn:indep}).
This condition permits flexible dependence among $\gamma_i$, $\beta_i$, and $U_i$, serial correlation in the errors, and dependence on $Y_{i0}$; it provides the source of variation to identify the full distribution of random coefficients and intercepts (rather than their means) in short panels with fixed $T$.
We impose no restriction on the covariance structure of the time-varying errors, which may be serially correlated and heteroskedastic in the initial condition $Y_{i0}$.
We show that both the joint distribution of $(\gamma_i,\beta_i,U_i)$ given $Y_{i0}$ and the coefficients for sequentially exogenous covariates $\delta_t$ are point identified.

Applying the analog principle to a constructive identification strategy, we propose a closed-form, multi-step estimator which does not involve any numerical optimization, simulation, or solution of an inverse problem beyond some deconvolution steps.
We establish uniform convergence of our estimator for the joint density of $(\gamma_i,\beta_i)$ and that of $U_i$ conditional on the random coefficients.

Our analysis departs from the literature on three fronts, which we preview here and develop in Section~\ref{sec:lit-review}. First, on the \emph{model}: we let the coefficient on the lagged dependent variable be individual-specific and identify its distribution under a \emph{fixed} $T$.
This is the regime in which \citet{Arellano2001} showed that even the mean of a heterogeneous autoregressive coefficient is under-identified when the covariates are only sequentially exogenous.
We circumvent this obstacle by bringing in strictly exogenous covariates under the distributional independence condition above, and we do so without any order condition relating $T$ to the number of covariates, and without the large-$T$ asymptotics that such panels often employ.

Second, on the \emph{target}: we recover the joint distribution of the random coefficients and time-varying errors, rather than a mean, an average partial effect, or a finite set of moments of the random coefficients alone, which the nearest dynamic random-coefficient panel literature targets, often via second-moment conditions, averaging, or partial identification.
In particular, we identify the distribution of the structural random intercepts, i.e., the time-varying errors in the \emph{structural} equation, and not merely that of the slope coefficients.
Triangular and simultaneous-equations models with random coefficients recover the coefficient distributions but leave that of the structural intercepts unidentified; we overcome this obstruction using a new argument that exploits the panel structure.

Third, on the \emph{method}: we adapt the inverse Radon transform used to estimate random-coefficient densities in static cross-sections \citep{hoderlein2010analyzing}, coupled with a deconvolution step, and apply them to an intermediate reduced-form density implied by our panel.
Recovering the structural objects from this reduced-form density is not immediate: it requires a Jacobian change of variables that extracts the joint density of $(\gamma_i,\beta_i)$, and an original argument, central to the proof of \Cref{theorem:two}, that identifies the distribution of $U_i$.
These steps constitute a large part of our methodological contribution.

\section{Related Literature} \label{sec:lit-review}

\noindent \textbullet \quad \textbf{Dynamic linear panel data models with random coefficients.}

Several seminal papers estimated dynamic linear panel data models with \emph{constant} coefficients under a fixed-$T$ setting.
These include \citet{AndersonHsiao1982}, \citet{HoltzEakin1988}, \citet{ArellanoBond1991}, \citet{ArellanoBover1995}, and \citet{BlundellBond1998}.
When the autoregressive coefficient is instead heterogeneous, identification becomes markedly harder: \citet{Arellano2001} showed that with only sequentially exogenous covariates and fixed $T$, even the \emph{mean} of that coefficient is under-identified. This motivates the strictly exogenous covariates and the distributional restriction we introduce in Section~\ref{subsec:iden_distr}.

Some papers studied dynamic linear panel models where strictly exogenous covariates have random coefficients but predetermined and sequentially exogenous ones have \textit{constant} coefficients.
\citet{ArellanoBonhomme2012} identified distributional features of the random coefficients for strictly exogenous covariates using second-moment conditions on time-varying errors.
In a forecasting framework, \citet{LiuMoonSchorfheide2020} constructed point predictors for the posterior mean of heterogeneous coefficients of strictly exogenous covariates and deterministic trends, under the assumption of Gaussian innovations and using Tweedie's formula.
We depart from these papers by letting the lagged dependent variable $Y_{it-1}$ carry a random coefficient, by dispensing with the Gaussian or second-moment conditions on the time-varying errors, and by recovering the joint distribution of random coefficients and intercepts.

\citet{Chamberlain2022} showed that in panel data models with sequentially exogenous covariates and a multi-dimensional vector of unobserved heterogeneity, the common finite-dimensional parameters (such as constant coefficients for other covariates) are not point identified.
\citet{Lee2026} studied the partial identification of a dynamic linear panel model with random coefficients for predetermined covariates and lagged dependent variables.
In contrast, we obtain point (rather than partial) identification of the joint distribution of random coefficients and intercepts by exploiting the full distribution (rather than a finite set of moments) of the outcomes.

Within a framework of potential outcomes and treatment effects, \citet{marx2025heterogeneous} obtain a causal interpretation of the classical IV/GMM estimands from the dynamic-panel literature \citep{AndersonHsiao1982, ArellanoBond1991} under a sequential exchangeability condition.\footnote{Sequential exchangeability is a joint restriction on the dynamics of potential outcomes and on treatment selection. \citet{marx2025heterogeneous} combine it with restrictions on treatment-effect heterogeneity to decompose the IV estimand into a convex-weighted aggregate of heterogeneous, history-dependent treatment effects.} They also relax that condition and identify certain average causal effects by imposing homogeneous autoregressive dynamics on the \emph{untreated} potential outcomes.
Our paper differs from \citet{marx2025heterogeneous} in two major ways. First, we study a distinct random-coefficient dynamic panel model and target a different parameter: the joint distribution of the random coefficients and time-varying errors in a model for \emph{observed} outcomes, rather than aggregates of treatment effects in a \emph{potential}-outcomes framework.
Second, we leverage a different source of variation. Our main identification results (Section~\ref{subsec:iden_distr}) hinge on a strong, distributional form of strict exogeneity of the covariates, whereas \citet{marx2025heterogeneous} include no strictly exogenous covariate by design, and instead use lagged instruments to handle covariates that are not strictly exogenous (e.g., treatments that depend on past observed outcomes).

\noindent \textbullet \quad \textbf{Multidimensional unobserved heterogeneity in linear panels.}

Other papers investigated linear panel models with vector, interactive, or correlated heterogeneity, typically focusing on the mean (instead of distribution) of such heterogeneity.
\citet{GrahamPowell2012} estimated average partial effects (or mean of random coefficients) in an ``irregular'' correlated random coefficient panel, using units whose covariates change little over time and requiring an order condition for just-identification (i.e., the number of periods $T$ equals the dimension of covariates).
Apart from the difference in the target parameter, their method does not apply in our case because the mean of composite random intercepts is not additive in a stationary function of the covariate history and a period-specific constant due to sequentially exogenous covariates.
Moreover, our method does not require any order condition on the number of time periods.

\citet{MoonWeidner2017} estimated a dynamic linear panel model that has a lagged dependent variable with a constant coefficient and interactive fixed effects.
\citet{LuSu2023, LuSu2025} studied linear panel models that allowed for random coefficients for sequentially exogenous covariates, but their theoretical frameworks do not include lagged dependent variables.
\citet{CaoJinLuSu2024} allowed for random coefficients on lagged dependent variables, but required a ``large-$N$, large-$T$'' setting where both the cross-section and time dimensions go to infinity.
In contrast with these papers, our method operates in a short panel ``fixed-$T$'' setting.
Furthermore, while these papers focus on estimating the mean (or a homogeneous/average partial effect) of random slope coefficients, we recover the joint distribution of the random coefficients on lagged dependent variables alongside time-varying errors.

Earlier and wider literature on random-coefficient panel models likewise targeted averages or low-order moments in static designs.
Examples include \citet{Swamy1970}, \citet{Chamberlain1992}, \citet{Wooldridge2005}, \citet{Murtazashvili2008}, and \citet{hsiao2008}.
\citet{Laage2024} estimated a correlated random-coefficient panel model with time-varying endogeneity, showing identification of the mean of random coefficients through control variables. Her model does not accommodate random coefficients for a lagged dependent variable.
\citet{li2026identification} estimated the average partial effect and the local average response in a correlated random coefficient panel data model, where regressors can be correlated with time-varying and individual-specific random coefficients.

\noindent \textbullet \quad \textbf{Triangular and simultaneous systems with random coefficients.}

Our dynamic linear panel model shares the structure of a triangular system, because $Y_{i1}$ feeds into $Y_{i2}$ as a predetermined regressor. Yet our paper differs from those on the triangular system, in terms of the target and the identification strategy. We review the core differences here, and relegate details of the comparison to Appendix~\ref{sec:Appendix_lit}.

\citet{HoderleinHolzmannMeister2017} identified the distribution of the random coefficients in a triangular model under two key restrictions: mutual independence between the coefficients and the instruments and exogenous covariates, and independence of the first-stage slopes from the outcome-equation coefficients (their Theorem 7). As we show in Appendix~\ref{sec:Appendix_lit}, these conditions are relaxed in our setting, and we additionally identify the distribution of random intercepts (time-varying errors) rather than that of the slope coefficients alone.

\citet{MastenTorgovitsky2016} identified an instrumental-variables correlated random coefficient model, relying on independence between the instruments and the coefficients together with a monotone relation between the endogenous covariate and a scalar latent control.
Our model does not fit their control-function framework: the first-stage equation is itself of random-coefficient form, which, as \citet{Imbens2007} and \citet{Kasy2011} showed, admits no such reduced-form, monotone representation with a scalar control.

\citet{masten2018random} treated triangular systems as a special case of simultaneous equations. He identified the joint distribution of the random coefficients for the endogenous regressor and the excluded instrument (his Proposition 4, Section 3.3), assuming the instrument is independent of all unobservables conditional on the shared covariates.
His analysis did not identify the distribution of the time-varying errors (random intercepts), nor that of the
random coefficients on the covariates shared in both equations.\footnote{
    Indeed, \citet{masten2018random} showed the joint distribution of all structural unobservables is not point identified, because the reduced form is a seemingly-unrelated-regression system with common regressors (his Theorem 4).}
In comparison, our panel setting avoids this obstruction, because each period supplies its own exogenous covariates: a suitable linear combination of the outcomes behaves as a single-equation random-coefficient model in \emph{distinct} regressors. Building on this, and varying the combination weights, we obtain a new result---identification of the full joint distribution of all random coefficients \textit{and} the time-varying
errors. This requires an extended argument that exploits the conditional independence between the exogenous covariates and the random coefficients and intercepts given the initial condition.

\noindent \textbullet \quad \textbf{Estimation of linear RC models via inverse Radon transform.}

\citet{Beran1996} and \citet{hoderlein2010analyzing} estimated the joint density of the random coefficients in static cross-section regressions by inverting the Radon transform (which links the conditional density of the outcome to the coefficient density); the latter did so through a kernel estimator with a special Radon-transform kernel and sharp asymptotics. We use this method for the first step in our estimation.

\section{The Model and Identification} \label{sec:iden}
We consider a dynamic panel data model:
\begin{equation} \label{eq:seq-exo-w}
    Y_{it} = \gamma_{i}Y_{it-1} + X_{it}'\beta_{it} + W_{it}'\delta_t + U_{it} \; \text{ for } t=1,2,\ldots,T,
\end{equation}
where \(\gamma_{i}\in\mathbb R\) and \(\beta_{it}\in\mathbb R^{J_x}\) are random coefficients, $ U_{it}\in\mathbb R$ are random intercepts, while $\delta_t\in\mathbb R^{J_w}$ are non-zero constant coefficients. There are no constant intercepts in $X_{it}, W_{it}$.
For a generic random array $\zeta_{it}$, let $ \zeta_i \equiv (\zeta_{i1}',\zeta_{i2}')'$ and $\zeta_i^t \equiv \{\zeta_{is}:s=1,\ldots,t\}$ denote the history up to time $t$, and $\Delta \zeta_{it} \equiv \zeta_{it} - \zeta_{it-1}$.
\medskip

\begin{assumption}[Strict and Sequential Exogeneity with Time Effects] \label{assn:seq-exo}
For $t=1,\ldots,T$,
    $$E(U_{it} \mid X_i^T,W_i^t,Y_{i0},\gamma_i) = \mu_t(Y_{i0}) < \infty \; \text{ almost surely,}$$ where $\mu_t(\cdot)$ is an unknown, unrestricted function.
\end{assumption} \medskip

Under this condition, $X_{it}$ is strictly exogenous while $W_{it}$ is sequentially exogenous.
It permits period-specific fixed effects $\mu_t(Y_{i0})$.
\citet{ArellanoBonhomme2012}
investigated a similar specification where the coefficients for strictly exogenous covariates are random but those for predetermined or sequentially exogenous covariates are all \emph{constant}, and $\mu_t(Y_{i0})=0$ for all $t$.
In comparison, our specification differs by allowing a predetermined regressor, namely, the lagged dependent variable $Y_{it-1}$, to have a random coefficient.

While \Cref{assn:seq-exo} restricts the conditional mean of $U_{it}$ given $\gamma_i,Y_{i0}$, it does allow dependence between $U_i$ and $\gamma_i$ through higher moments (as in the correlated and scale designs in our simulation study).

The first part of our identification results (Section \ref{subsec:iden_moment}) does not require conditions on the second moments of time-varying errors used by \citet{ArellanoBonhomme2012}.
On the other hand, our method does rely on a stronger notion of strict exogeneity, i.e., the distributional independence in \Cref{assn:indep} instead of mean independence (Assumption 1) in \citet{ArellanoBonhomme2012}.

The parameters of interest are the joint distribution of random coefficients and intercepts $(\gamma_i, \beta_{i1}',\ldots,\beta_{iT}', U_{i1},\ldots,U_{iT})$ and the constant coefficients $(\delta_1,\ldots,\delta_T)$.
We prove point identification of these parameters via sequential steps.

\subsection{The conditional mean of random coefficients} \label{subsec:iden_moment}

First, we identify the constant coefficients $\delta_t$ for $t=1,\ldots,T$ and the conditional mean of random coefficients.
To fix ideas, focus on the case with $T=2$.
By recursive substitution and first-differencing,
\begin{align*}
    \begin{split}
     \Delta Y_{i2} & \equiv  Y_{i2} - Y_{i1}  \\
    & =  \gamma_i(\gamma_i-1)Y_{i0} + (\gamma_i-1)X'_{i1} \beta_{i1}  +
       X'_{i2} \beta_{i2} + (\gamma_i-1)W'_{i1} \delta_1 + W'_{i2}\delta_2 + \varpi_i,
    \end{split}
\end{align*}
where \(\varpi_i \equiv \Delta U_{i2} + \gamma_iU_{i1}\).
Let $H_i$ be shorthand for the covariates that are exogenous with respect to $U_i\equiv (U_{i1},U_{i2})'$ conditional on the initial condition $Y_{i0}$. That is:
\[H_{i}\equiv (X_{i1}',X_{i2}',W_{i1}')'.\]
We also maintain a conditional mean independence condition on the random coefficients.
Let $\theta_i \equiv (\gamma_i,\beta_i')'$ with $\beta_i\equiv (\beta_{i1}',\beta_{i2}')'$.

\medskip

\begin{assumption} \label{assn:mean-indep}
$E(\theta_i\mid H_i,Y_{i0}) = E(\theta_i\mid Y_{i0})$ and $E(\theta_i\gamma_i\mid H_i, Y_{i0}) = E(\theta_i\gamma_i \mid Y_{i0})$.
\end{assumption} \medskip

Under Assumptions \ref{assn:seq-exo} and \ref{assn:mean-indep},
\[ E( \varpi_i \mid H_{i}, Y_{i0} ) = E\big[\gamma_i\mu_1(Y_{i0}) + \Delta \mu(Y_{i0}) \mid H_i, Y_{i0} \big] = m(Y_{i0}), \]
where $m(Y_{i0})\equiv E(\gamma_i\mid Y_{i0})\mu_1(Y_{i0}) + \Delta \mu(Y_{i0})$ with $ \Delta \mu \equiv \mu_2 - \mu_1 $, and
\begin{equation} \label{eq:levelreg}
    E(Y_{i1}\mid H_i,Y_{i0}) = \pi_1(Y_{i0}) + X_{i1}'E(\beta_{i1}\mid Y_{i0}) + W_{i1}'\delta_1,
\end{equation}
with \(\pi_1(Y_{i0}) \equiv E(\gamma_i\mid Y_{i0})Y_{i0} + \mu_1(Y_{i0})\), and
\begin{equation} \label{eq:delta-y2}
    E(\Delta Y_{i2} \mid H_i, Y_{i0}) = Z_i'\Psi(Y_{i0}),
\end{equation}
where $Z_i\equiv(\,1,H_i',E(W_{i2}'\mid H_i,Y_{i0})\,)'$ and
\begin{align} \label{defn:Psi}
    \Psi(Y_{i0}) \equiv
        \begin{pmatrix}
           \psi_1(Y_{i0}) \\
           \psi_2(Y_{i0}) \\
           \psi_3(Y_{i0}) \\
           \psi_4(Y_{i0}) \\
           \psi_5(Y_{i0})
        \end{pmatrix}
        \equiv
        \begin{pmatrix}
           E(\gamma_i^2 - \gamma_i\mid Y_{i0}) Y_{i0} + m(Y_{i0}) \\
           E[\beta_{i1}(\gamma_i - 1) \mid Y_{i0}] \\
            E(\beta_{i2}\mid Y_{i0}) \\
           \delta_1E(\gamma_i-1\mid Y_{i0}) \\
           \delta_2
         \end{pmatrix}.
\end{align}
That is, the means of $Y_{i1}$ and $\Delta Y_{i2}$ conditional on $H_i,Y_{i0}$ are both linear in $Z_i$. We maintain the following condition on the (conditional) support of $Z_i$. \medskip

\begin{assumption}[Non-singularity] \label{assn:rank}
    $E(\,\| Z_i \|^2\mid Y_{i0}\,)< \infty$, and the support of $Z_i$ conditional on $Y_{i0}$ is not contained in any proper linear subspace of $\mathbb R^{2(J_x+J_w)+1} $ almost surely.
\end{assumption} \medskip

A sufficient condition for \Cref{assn:rank} is that
for any nonzero $c\in\mathbb R^{J_w}$, $c'E(W_{i2}\mid H_i,Y_{i0})$ is not almost surely affine in $H_i$ given $Y_{i0}$.
That is, nonlinearity in the conditional mean of $W_{i2}$ is essential.
Under \Cref{assn:rank}, both $E(Z_iZ_i'\mid Y_{i0})$ and $E(R_iR_i'\mid Y_{i0})$ are non-singular almost surely, where $R_i\equiv(1,X_{i1}',W_{i1}')'$ is a sub-vector of $Z_i$.
Therefore, we can identify $\pi_1(Y_{i0})$, $E(\beta_{i1}\mid Y_{i0})$, and $\delta_1$ from \eqref{eq:levelreg} using variation in $R_i$ given $ Y_{i0} $, and identify $\Psi(Y_{i0})$ from \eqref{eq:delta-y2} using variation in $Z_i$ given $Y_{i0}$.
It then follows from \eqref{defn:Psi} that
\begin{align}
    & E(\beta_{i2}\mid Y_{i0}) = \psi_3(Y_{i0}),\quad \delta_2 = \psi_5(Y_{i0}), \quad E(\gamma_i\mid Y_{i0}) = \delta_1'\psi_4(Y_{i0})/\|\delta_1\|^2 + 1 , \\
    & \mu_1(Y_{i0}) = \pi_1(Y_{i0}) - E(\gamma_i\mid Y_{i0})Y_{i0}, \quad
    E(\gamma_i\beta_{i1}\mid Y_{i0}) = \psi_2(Y_{i0}) + E(\beta_{i1}\mid Y_{i0}).
    \nonumber
\end{align}
The proposition below collects these identification results based on the mean of $Y_{it}$. \medskip

\begin{proposition} \label{pn:iden-rc-means}
    Under Assumptions \ref{assn:seq-exo}, \ref{assn:mean-indep} and \ref{assn:rank}, $\delta_1$, $\delta_2$, $E(\beta_{i1}\mid Y_{i0})$, $E(\beta_{i2}\mid Y_{i0})$, $E(\gamma_i\mid Y_{i0})$, $E(\gamma_i\beta_{i1}\mid Y_{i0})$ and $\mu_1(Y_{i0})$ are identified from $E(Y_{it}\mid H_i,Y_{i0})$ for $t=1,2$.
\end{proposition} \medskip

\begin{remark}
    The coefficients for sequentially exogenous covariates $\delta_1,\delta_2$ are all we need for the next step in \Cref{subsec:iden_distr}.
    Nevertheless, we note the identification of the other parameters in \Cref{pn:iden-rc-means} is also useful, as they are obtained from conditional mean outcomes without invoking the distributional form of strict exogeneity in \Cref{assn:indep}.
\end{remark}

\begin{remark}
    The second moment $E(\gamma_i^2\mid Y_{i0})$ and $\mu_2(Y_{i0})$ enter additively in $\psi_1(Y_{i0})$, and cannot be separately identified from the conditional means without further assumptions. We propose two approaches to identify them separately.
    The first is to use a stationarity condition that $ \mu_t(Y_{i0}) $ is identical over $t=1,2$ almost surely. This allows us to recover $E(\gamma_i^2\mid Y_{i0})$ from $\psi_1(Y_{i0})$ using knowledge of parameters identified in \Cref{pn:iden-rc-means}.
    The second approach is to strengthen the mean independence of random intercepts and coefficients in Assumptions \ref{assn:seq-exo} and \ref{assn:mean-indep} to a stronger form of distributional independence (\Cref{assn:indep}), which allows us to recover the distribution (and the second moment) of $\gamma_i$ given $Y_{i0}$ as in \Cref{theorem:two}. We can then recover $\mu_2(Y_{i0})$ from $\psi_1(Y_{i0})$, again using results from \Cref{pn:iden-rc-means}.
\end{remark}

\begin{remark}
    If we strengthen \Cref{assn:mean-indep} with $E(\theta_i\mid Y_{i0})=E(\theta_i)$, $E(\gamma_i^2\mid Y_{i0})=E(\gamma_i^2)$, and $\mu_t(Y_{i0})=\mu_t$, then the conditional means of $Y_{i1}$ and $\Delta Y_{i2}$ are both \emph{unconditionally} linear in $H_i$ \emph{and} $Y_{i0}$. The unconditional expectations of $U_i,\gamma_i,\beta_i$ as well as $\delta_1,\delta_2$ are identified from simple regressions of $Y_{i1},\Delta Y_{i2}$ that pool over $Y_{i0}$ and include it as a regressor along with $H_i$.
\end{remark}

\begin{remark}
    Our method does not need $T$ to be at least as large as the dimension of strictly exogenous covariates; this differs from \citet{ArellanoBonhomme2012} (with $T$ strictly larger than the latter) and \citet{GrahamPowell2012} (with $T$ equal to the latter).
    Furthermore, our model differs from these two papers by allowing the distribution of the random coefficients for strictly exogenous covariates to vary over time, and not to be additive in a stationary function of covariate history. (See Appendix \ref{sec:Appendix_lit} for details.)
\end{remark}

    \Cref{pn:iden-rc-means} uses conditional mean outcomes for identification, and invites comparison with Section 3 of \citet{marx2025heterogeneous}.
    Beyond the difference in target parameters, their identifying assumptions and ours are non-nested. To compare the two, map their treatment $D_{it}$ into our sequentially exogenous covariate $W_{it}$, and their composite error $\theta_t + \alpha_i + \varepsilon_{it}$ into our random intercept $U_{it}$.
    First, \citet{marx2025heterogeneous} focus on the IV/GMM estimand and include \emph{no} strictly exogenous covariate, whereas we require at least one.
    Second, sequential exchangeability in \citet{marx2025heterogeneous} permits past outcomes to feed into the current treatment, carrying information about the fixed effects beyond $Y_{i0}$ and resulting in $E(U_{i2} \mid D_{i1}, D_{i2}, Y_{i0}) \neq \mu_2(Y_{i0})$ in general, thus violating our \Cref{assn:seq-exo}.
    Moreover, the two approaches draw on different variation: we exploit the (conditional) linearity of $E(Y_{it} \mid H_i, Y_{i0})$ in exogenous covariates, whereas \citet{marx2025heterogeneous} use lagged outcomes as instruments in a first-differenced equation.

\subsection{Distribution of random coefficients and intercept} \label{subsec:iden_distr}

Next, we identify the joint distribution of random coefficients and intercepts, using a stronger form of conditional independence.
Throughout this section, we treat $\delta_1,\delta_2$ as known, having been identified in \Cref{pn:iden-rc-means}.
\medskip

\begin{assumption} \label{assn:indep}
    \((\gamma_i,\beta_i,U_i) \perp H_i \mid Y_{i0}\),
    and the conditional distribution of $(\gamma_i,\beta_i,U_i)$ given $ Y_{i0}$ admits a Lebesgue density almost surely.
\end{assumption} \medskip

\Cref{assn:indep} allows flexible dependence among $\gamma_i$, $\beta_i$ and $U_i$, and is compatible with Assumptions \ref{assn:seq-exo} and \ref{assn:mean-indep}.
It also allows the distribution of random coefficients and intercepts to be heterogeneous in the initial condition $Y_{i0}$.
In the special case with additive fixed effects $U_{it} = \alpha_i + \varepsilon_{it}$, \Cref{assn:indep} accommodates general dependence between $Y_{i0},\gamma_i,\beta_i,\alpha_i$ and $\varepsilon_i\equiv(\varepsilon_{i1},\varepsilon_{i2})$, as well as serial correlation of $\varepsilon_{it}$.
\medskip

Define \(\wt Y_{it} \equiv Y_{it} - W'_{it}\delta_t \) for \( t = 1, 2\).
For any $d\equiv(d_1,d_2)\in\mathbb R^2$,
\begin{align}
\label{eq:ty}
    d_1\wt Y_{i1} + d_2\wt Y_{i2}
     =  C_i(d) + X_{i1}'\underset{S_{i1}(d)}{\underbrace{(\beta_{i1}d_1 + \beta_{i1}\gamma_id_2)}} +
           X_{i2}'\underset{S_{i2}(d)}{\underbrace{\beta_{i2}d_2}} + W_{i1}'\underset{S_{i3}(d)}{\underbrace{\delta_1\gamma_id_2}},
\end{align}
where \[C_i(d) \equiv (d_1\gamma_i + d_2\gamma_i^2)Y_{i0} + (d_1+d_2\gamma_i)U_{i1} + d_2U_{i2}.\]
\noindent Let $S_i(d) \equiv (S_{ik}(d):k=1,2,3)$.
Under \Cref{assn:indep},
\begin{equation} \label{eq:cond-indep}
    (\;C_i(d),S_{i}(d)\;)\perp H_i \mid Y_{i0} \text{ for any }d\in\mathbb R^2.
\end{equation}
The {\it first} step in identifying the distribution of random coefficients is to recover the joint distribution of $(\,C_i(d),S_i(d)\,)$ conditional on $Y_{i0}$. \medskip

\begin{assumption}\label{assn:richSupp}\leavevmode
\begin{enumerate}[(i)]
    \item \label{assn:richSupp:i} The support of \(H_i\mid Y_{i0}\) contains an open ball in $\mathbb R^{2J_x+J_w}$ almost surely.
    \item \label{assn:richSupp:ii}   For every $d\in\mathbb R^2$, the conditional distribution of $(C_i(d), S_i(d))$ given $Y_{i0}$ is uniquely determined by its moments, and has finite absolute moments of all orders almost surely.
\end{enumerate}
\end{assumption} \medskip

For linear regressions with random coefficients and intercept in reduced form, \citet{masten2018random} (Lemma 2) showed that the moment-determinacy and finite-moment conditions such as (\ref{assn:richSupp:ii}) are sufficient for point-identification, provided these are independent of the regressors whose joint support satisfies an open-ball condition as in (\ref{assn:richSupp:i}).\footnote{
    The generic pair $(A,B)$ and regressors $Z$ in \citet{masten2018random} correspond to $ (C_i(d),S_i(d)) $ and $H_i$ in our case, respectively. \citet{masten2018random} also showed that if $\mathrm{supp}(H_i\mid Y_{i0})$ is bounded, then the conditions in (\ref{assn:richSupp:ii}) are necessary for identification.}

A primitive sufficient condition for \eqref{assn:richSupp:ii} is that, conditional on $Y_{i0}$, the components of $(\gamma_i,\beta_i,U_i)$ have sub-Gaussian tails, which permits them to have non-compact support.
Every linear combination of $(C_i(d),S_i(d))$, whose entries include the quadratic term $\gamma_i^2 Y_{i0}$ as well as the products $\gamma_i U_{i1}$ and $\gamma_i \beta_{i1}$, is then sub-exponential
and hence satisfies Carleman's condition, ensuring moment determinacy of the joint conditional distribution for every $d \in \mathbb{R}^2$ \citep{Petersen1982}.

The following lemma follows immediately from the conditional independence we established in (\ref{eq:cond-indep}), \Cref{assn:richSupp}, and Lemma~2 of \citet{masten2018random}.

\medskip

\begin{lemma}\label{lm:one}
    Suppose Assumptions \ref{assn:indep} and \ref{assn:richSupp} hold. For all $d\in\mathbb R^2$, the distribution of $(C_i(d),S_i(d))$ given $Y_{i0}$ is identified from that of $(d_1\wt Y_{i1} + d_2\wt Y_{i2},H_i')$ given $Y_{i0}$.
\end{lemma}

\medskip

The {\it second} step is to recover the joint density of $(C_i(d),\gamma_i,\beta_i)$ given $Y_{i0}$ from that of $(C_i(d),S_i(d))$ using Jacobian transformation. Specifically, choose $j\in\{1,\ldots,J_w\}$ so that the $j$-th component of $\delta_1$, denoted by $\delta_{1,j}$, is nonzero, and define a Jacobian matrix:
\begin{equation} \label{defn:Jacobian}
J_i(d) \equiv
\begin{pmatrix}
    1 & 0 & 0 & 0 \\
    0 & d_2\beta_{i1} & (d_1+d_2\gamma_i)I & 0 \\
    0 & 0 & 0 & d_2I \\
    0 & d_2\delta_{1,j} & 0 & 0
\end{pmatrix}, \ \text{ where } I \text{ is a $J_x$-by-$J_x$ identity matrix.}
\end{equation}
Suppose for $d$ with $d_2\neq 0 $, the determinant of $J_i(d)$ is non-zero almost surely.
(Because the conditional distribution of $\gamma_i$ given $Y_{i0}$ is atomless under \Cref{assn:indep}, $\Pr\{d_1+d_2\gamma_i=0 \mid Y_{i0}\} = 0 $ for any $d_2\neq 0$ almost surely.)
Recover the joint density of $C_i(d),\gamma_i,\beta_i$ conditional on $Y_{i0}=y_0$ as:
\begin{align}\label{eq:iden_CRC}
    \begin{split}
    &f_{C_i(d),\gamma_i,\beta_{i1},\beta_{i2} \vert Y_{i0}=y_0}(c,\gamma,b_1,b_2) \\
     = & f_{C_i(d),S_{i1}(d),S_{i2}(d),S_{i3,j}(d)\vert Y_{i0}=y_0}(c,b_1(d_1+d_2\gamma),b_2d_2,d_2\gamma\delta_{1,j}) \times
          \vert \delta_{1,j}d^{J_x+1}_2(d_1 + d_2\gamma)^{J_x}\vert.
    \end{split}
\end{align}
This identifies the joint distribution of $(\gamma_i,\beta_i)$ given $Y_{i0}$.

In the {\it third} step, for any $a \equiv (a_1,a_2) \in \mathbb R^2$ and realization of random coefficients $\gamma$, define the map \(\bar d(a,\gamma) \equiv (a_1-a_2\gamma,\,a_2)\), with components \(\bar d_1 \equiv a_1-a_2\gamma\) and \(\bar d_2 \equiv a_2\). Then for any $b\equiv(b_1',b_2')'$ and initial condition $y_0$,
\begin{align}\label{eq:iden_tU}
    \begin{split}
         & C_i(\bar d) \mid \gamma_i = \gamma, \beta_i=b, Y_{i0} = y_0  \\
     \sim & \ (\bar d_1 + \bar d_2\gamma)U_{i1} + \bar d_2 U_{i2} + \psi^* \mid \gamma_i = \gamma, \beta_i=b, Y_{i0} = y_0  \\
     \sim & \ a_1U_{i1} + a_2 U_{i2} + \psi^* \mid \gamma_i = \gamma, \beta_i=b, Y_{i0} = y_0,
    \end{split}
\end{align}
where \(\psi^* \equiv  (\bar d_1 \gamma + \bar d_2\gamma^2) y_0 = \gamma a_1 y_0\) is a known constant.

{\it Lastly}, combine the three steps above and note for any non-zero $a_1,a_2\in\mathbb R$, the distribution of \(a_1U_{i1} + a_2U_{i2}\) conditional on \((\gamma_i,\beta_i,Y_{i0})\) is identified from that of \((\wt Y_{i1}, \wt Y_{i2}, H_i')\) given $Y_{i0}$.
Therefore, the characteristic function (and the distribution) of \(U_i=(U_{i1},U_{i2})'\) given $(\gamma_i,\beta_i,Y_{i0})$ is identified.\footnote{
    We show identification of the characteristic function over nonzero $(a_1,a_2)$. Recovery of the characteristic function at $a$ s.t. $a_1a_2=0$ follows from uniform continuity of characteristic functions.}
In particular, the intercept means $\mu_t(Y_{i0})=E(U_{it}\mid Y_{i0})$ and $E(\gamma_i^2\mid Y_{i0})$---left non-separable by the conditional-mean step---are recovered here as features of this joint law.
The next theorem formalizes the identification result.

\medskip

\begin{theorem}\label{theorem:two}
    Suppose Assumptions \ref{assn:seq-exo}, \ref{assn:mean-indep}, \ref{assn:rank}, \ref{assn:indep} and \ref{assn:richSupp} hold. Then $\delta_1,\delta_2$ and the joint distribution of \((U_i,\gamma_i,\beta_i)\) given \(Y_{i0}\) are identified.
\end{theorem}

\medskip

\begin{remark}
    \Cref{theorem:two} uses \Cref{assn:seq-exo}--\ref{assn:rank} only to identify $\delta_1,\delta_2$ via \Cref{pn:iden-rc-means}.
    Given $\delta_1,\delta_2$, the distributional argument rests on \Cref{assn:indep}--\ref{assn:richSupp} alone.
    \Cref{assn:mean-indep} is implied by \Cref{assn:indep} (as long as the conditional means exist).
\end{remark}

The identification strategy above is constructive; we apply the analog principle to define closed-form estimators in the next section.

In Appendix~\ref{sec:T3}, we extend this identification strategy to any fixed $T\ge3$ even without variation in sequentially exogenous $W_{it}$.

\section{Closed-Form Estimators}\label{sec:est}

We propose multi-step estimators for $\delta_1,\delta_2$ and the joint density of $ (\gamma_i, \beta_i, U_i)$.
The identification results are conditional on the initial condition $Y_{i0}$; we fix it at a constant $y_0\neq 0$ throughout this section, and suppress it from the notation. As in Section \ref{sec:iden}, we focus on the case with $T=2$.

\subsection{Moments of random coefficients and \texorpdfstring{$\delta_t$}{deltat}}
\label{subsec:est-delta}

Let \(T=2\) and use lower-case letters to denote the realizations in the sample, e.g., $h_i \equiv (x_{i1}',x_{i2}',w_{i1}')'$.
Let $\wh w_{i2}$ denote a Nadaraya-Watson estimator for the mean of $W_{i2}$ conditional on $h_i$;
let $\widehat z_i \equiv (1,h_i',\wh w_{i2}')'$, and
\[\widehat \Psi \equiv \left(\sum\nolimits_i \wh z_i\wh z_i'\right)^{-1}
\left(\sum\nolimits_i\wh z_i\Delta y_{i2}\right),\]
where the $k$-th component $\widehat \psi_k$ is an estimator of $\psi_k$ in (\ref{defn:Psi}), with the initial condition $y_0$ fixed and suppressed in notation.
Let $r_i \equiv (1,x_{i1}',w_{i1}')'$ and estimate the period-1 level regression \eqref{eq:levelreg} by
\begin{equation} \label{eq:est_condY1}
    \wh\Phi \equiv \left(\sum\nolimits_i r_ir_i'\right)^{-1}\left(\sum\nolimits_i r_i\, y_{i1}\right) = (\wh \pi_1,\ \wh E(\beta_{i1})',\ \wh\delta_1')'.
\end{equation}
Together with $\wh\Psi$ from the regression on $\Delta Y_{i2}$, we also estimate
\begin{align} \label{est_condMean}
   \wh E(\beta_{i2}) \equiv & \ \wh\psi_3, \quad \wh\delta_2 \equiv \wh\psi_5, \quad
   \wh E(\gamma_i) \equiv 1 + \wh\delta_1'\wh\psi_4/\|\wh\delta_1\|^2.
\end{align}
We can also estimate the following parameters (even though they are not needed for subsequent estimation of the joint density of $(\gamma_i,\beta_i',U_i')'$):
\begin{align*}
   & \wh\mu_1 \equiv \wh \pi_1 - \wh E(\gamma_i)\,y_0, \quad
     \wh E(\gamma_i^2) \equiv \int \gamma^2\,\wh f_{\gamma_i}(\gamma)\,d\gamma, \quad \wh E(\gamma_i\beta_{i1})\equiv \wh\psi_2+\wh E(\beta_{i1}), \nonumber \\
   & \wh\mu_2 \equiv \wh\psi_1 - \big(\wh E(\gamma_i^2)-\wh E(\gamma_i)\big)y_0 - \wh E(\gamma_i)\wh\mu_1 + \wh\mu_1, \nonumber
\end{align*}
where $\wh f_{\gamma_i}(\gamma) \equiv \int \wh f_{\gamma_i,\beta_i}(\gamma,b)\,db$ uses the density estimator of Section \ref{subsec:est_RC}.

With $\delta_1\neq 0$, the probability limit of the denominator $\|\wh\delta_1\|^2$ in $\wh E(\gamma_i)$ is strictly bounded away from zero.
Consequently, a standard argument using Slutsky's Theorem ensures that division by $\|\wh\delta_1\|^2$ does not affect the $\sqrt{n}$-consistency and asymptotic normality of the estimators in \eqref{eq:est_condY1} and \eqref{est_condMean}.

\subsection{The joint density of random coefficients} \label{subsec:est_RC}

\noindent
For the rest of this section and Section~\ref{sec:asymp}, assume $W_{i1}\in\mathbb R$ for simplicity.
Fix $d=(0,1)$, so that (\ref{eq:ty}) reduces to
\begin{equation*}
    \wt Y_{i2} = \underbrace{\gamma_i^2 y_0 + \gamma_iU_{i1}+U_{i2}}_{C_i^*}+ X'_{i1}\underbrace{\beta_{i1}\gamma_i}_{S_{i1}^*}+X'_{i2}\underbrace{\beta_{i2}}_{S_{i2}^*}+W'_{i1}\underbrace{\delta_1\gamma_i}_{S_{i3}^*}.
\end{equation*}
Following \cite{hoderlein2010analyzing}, we define an (inverse) Radon transform estimator (RTE) for the density of $(C_i^*,S_i^*)$.

Recall $ H_i \equiv (X_{i1}',X_{i2}',W_{i1}')' $. For $i=1,\dots,n$, define
\begin{equation*}
    Q_i \equiv \| (1, H_i')\|^{-1}(1, H_i')' \in \mathbb{S}_+^{l-1} \ \text{ and } \
    V_i \equiv \|(1, H_i')\|^{-1} \left(Y_{i2} - W_{i2}'\wh\delta_2\right) \in \mathbb{R},
\end{equation*}
where $\|\cdot\|$ denotes the Euclidean norm, and $\mathbb{S}_+^{l-1} = \{z \in \mathbb{R}^l: z_1>0, \|z\|=1 \}$ is the upper hemisphere of the unit sphere in $\mathbb{R}^l$ with $l\equiv\dim[(1, H_i')]=2J_x+2$.

Define \(
\mathbb{S}(\underline q_n)\equiv\{q\in\mathbb{S}_+^{l-1}: q_1\ge \underline q_n\}\), where $\underline q_n\to 0$ as $n\to\infty$.
Estimate the density of $(C_i^*,S_i^*)$ by
   \begin{equation} \label{eq:est_cstar}
       \wh f_{C_i^*,S_i^*}(c,s) \equiv \frac{2}{n} \sum\nolimits_i \frac{\mathbf 1\{q_{i1}\ge \underline q_n\}}{\wh f_{Q_i}(q_i)}K_\nu(q_i'(c,s')'-v_i),
    \end{equation}
where $\mathbf 1(\cdot)$ is the indicator function, and $K_\nu$ is the kernel defined by Equations (7)--(8) of \cite{bissantz2014confidence}.\footnote{\label{note:defn_K}
    The kernel is defined through its Fourier transform $\mathcal F K_\nu(t)=\frac{1}{2}(2\pi)^{-l+1}|t|^{l-1}\mathcal L(\nu|t|)$, which by symmetry of $\mathcal L$ admits an explicit form:
    \(K_\nu(w)\equiv (2\pi)^{-l}\int^\infty_0\cos(tw)\,t^{l-1}\mathcal L(\nu t)\,dt\).
    The symmetric filter is \(\mathcal L(t)\equiv(1-|t|^r)\mathbf{1}_{[-1,1]}(t)\), with an order parameter $0<r<\infty$.
    The kernel $K_\nu$ can be written, via a change of variable $\tilde t=\nu t$, as \(K_\nu(w)=\nu^{-l}(2\pi)^{-l}\!\int_0^\infty\!\cos\!\big(\tilde t\,\tfrac{w}{\nu}\big)\tilde t^{\,l-1}\mathcal L(\tilde t)\,d\tilde t=\nu^{-l}K(\tilde w)\big|_{\tilde w=w/\nu}\), where \( K(\tilde w)\equiv (2\pi)^{-l}\!\int_0^\infty\!\cos(t\tilde w)\,t^{l-1}\mathcal L(t)\,dt \). }
The denominator is a kernel estimator of the density of $Q_i$ on its spherical support:
\begin{equation*}
        \wh f_{Q_{i}}(\tilde q) \equiv \frac{1}{n}\sum\nolimits_i \varkappa(\tau)\wt K\left(\tau^{-2}(1-\tilde q'q_i)\right),
    \end{equation*}
where $\wt K(\cdot)$ is a kernel function, $\tau > 0$ is a smoothing parameter, and $\varkappa(\tau)$ is a normalization constant defined in \cite{hoderlein2010analyzing}.

Estimate the density of $(C_i^*,\gamma_i,\beta_i)$ by
    \[\wh f_{C_i^*,\gamma_i,\beta_i}(c,\gamma,b) \equiv \wh f_{C_i^*,S_i^*}(c, b_1\gamma,b_2,\wh\delta_1\gamma) \vert \wh\delta_1\gamma^{J_x}\vert,\]
and estimate the density of $(\gamma_i,\beta_i)$ by
\[\wh f_{\gamma_i,\beta_i}(\gamma,b) \equiv \int_{\mathcal C_n} \wh f_{C_i^*,\gamma_i,\beta_i}(c,\gamma,b)\, d c ,\]
where $\mathcal C_n\equiv[-\bar c_n,\bar c_n] \subset \mathbb{R}$, with $\bar c_n\to\infty$ as $n\to\infty$.

\subsection{The joint density of \texorpdfstring{$U_i$}{Ui} conditional on random coefficients} \label{subsec:est_U}

For any $a = (a_1,a_2) \in \mathbb{R}^2$ and any given $\gamma$ on the support of $\gamma_i$, define $\bar d \equiv (a_1-a_2\gamma, a_2)$.
With slight abuse of notation, write $C_i(\bar d) = C_i(a,\gamma)$ and $S_i(\bar d) = S_i(a,\gamma)$. Define:
\[V_i(a,\gamma) \equiv \|(1, H_i')\|^{-1} [\ (a_1-\gamma a_2)\wh Y_{i1} + a_2 \wh Y_{i2}\ ], \]
where $\wh Y_{it} \equiv Y_{it} - W_{it}'\wh\delta_t$ for $t=1,2$.
Estimate the density of $(C_i(a,\gamma), S_i(a,\gamma))$ by replacing $V_i$ with $V_i(a,\gamma)$ in (\ref{eq:est_cstar}), and denote that estimator by $\wh f_{C_i(a,\gamma), S_i(a,\gamma)}$.
For simplicity, we suppress the arguments $(a,\gamma)$ for the rest of this subsection.

Estimate the density of $(C_i, \gamma_i, \beta_i)$ at $(c,\gamma,b)$ by
\begin{align*}
  \wh f_{C_i, \gamma_i, \beta_{i}} (c,\gamma,b) \equiv \wh f_{C_i, S_{i1}, S_{i2}, S_{i3}}(c,b_1a_1,b_2a_2, \wh\delta_1a_2\gamma)|\wh\delta_1a_1^{J_x}a_2^{J_x+1}|,
\end{align*}
and estimate the conditional density of $C_i$ given $\gamma_i = \gamma, \beta_i=b$ as
\begin{align*}
    \wh f_{C_i\mid \gamma_i = \gamma, \beta_i=b}(c) = \frac{\wh f_{C_i, \gamma_i, \beta_{i}} (c,\gamma,b)}{\wh f_{\gamma_i,\beta_i}(\gamma,b)}.
\end{align*}

For any $a\in\mathbb R^2$, let $\wt U_i(a) \equiv C_i - \psi^* = a_1U_{i1}+a_2U_{i2}$, where $\psi^*=a_1\gamma y_0$ as defined in (\ref{eq:iden_tU}).
Estimate the conditional density of $\wt U_i(a)$ by
\begin{align*}
    \wh f_{\wt U_i(a) \mid \gamma_i = \gamma, \beta_i=b}(\tilde u)
    = \wh f_{C_i \mid \gamma_i = \gamma, \beta_i=b}(\tilde u+a_1\gamma y_0).
\end{align*}

Estimate the joint density of $U_i=(U_{i1},U_{i2})'$ conditional on $\gamma_i = \gamma, \beta_i=b$ following steps similar to those in \cite{masten2018random}.
First, estimate the conditional characteristic function of $U_i$ by
\begin{equation*}
\wh \phi_{U_{i1},U_{i2}\mid \gamma_i = \gamma, \beta_i=b}(a_1,a_2)=\int_{\mathcal U_n} \exp(i\tilde u)\, \wh f_{\wt U_i(a) \mid \gamma_i = \gamma, \beta_i=b}(\tilde u)\, d\tilde u,
\end{equation*}
where $\mathcal U_n$ is the interval implied by $\bar c_n$, which expands as $n\to\infty$; see the proof of \Cref{theorem:u} in Appendix~\ref{sec:AppendixA}.

Then, use the inverse Fourier transform to estimate the conditional density of $U_i$ as
\begin{align*}
&\wh f_{U_{i1},U_{i2}\mid \gamma_i= \gamma, \beta_i=b}(u_1,u_2) \\
= &\mathrm{Re}\!\left(\frac{1}{(2\pi)^2} \int_{\mathcal{A}_n}\!\exp[-i(a_1u_1+a_2u_2)]\wh \phi_{U_{i1},U_{i2}\mid \gamma_i = \gamma, \beta_i=b}(a_1,a_2)da_1da_2\right),
\end{align*}
where $\mathrm{Re}(\cdot)$ is the real part of a complex number and
\[\mathcal{A}_n\equiv\{a\in\mathbb R^2:\ \underline{\vartheta}_n\le|a_1|\le \bar{\vartheta}_n,\ \underline{\vartheta}_n\le|a_2|\le \bar{\vartheta}_n\},\]
with tuning sequences $\bar{\vartheta}_n\to\infty$ and $\underline{\vartheta}_n\to0$ as $n\to\infty$.
Truncating the Fourier inversion to a bounded region is standard in deconvolution \citep{fan1991optimal,meister2009deconvolution}; \citet{masten2018random} applied the same technique in estimation.

\section{Asymptotic Properties} \label{sec:asymp}

Define a Sobolev space \[\mathcal{W}^m(\mathbb{R}^l) \equiv \{f \in L^2(\mathbb{R}^l): (1+\|\cdot \|^2)^{m/2}\mathcal{F}f(\cdot) \in L^2(\mathbb{R}^l)\}\] with a semi-norm
\[\|f\|^2_m \equiv \int_{\mathbb R^l}(1+\|\xi\|^2)^m | \mathcal F f(\xi)|^2d\xi,\]
where $\mathcal{F}f(\xi)\equiv \int_{\mathbb R^l}f(x) \exp(-ix'\xi)dx$ is the Fourier transform of $f(\cdot)$, and $m>0$ is the {\it order of smoothness} of the Sobolev space.
For any constant $B>0$, define
\[\mathcal{W}^m(\mathbb{R}^l; B) \equiv \{ f \in \mathcal{W}^m(\mathbb{R}^l): \| f\|_m \leq B\}.\]
We treat $\delta_t$ as known.
The first-step estimators $\wh \delta_1,\wh \delta_2$ in Section \ref{subsec:est-delta} are $\sqrt n$-consistent (with $Y_{i0}$ fixed at a constant $y_0$ in the sample).
With $m > l/2$, the uniform convergence rate of the intermediate density estimator is strictly slower than $n^{-1/2}$ (\Cref{lemma:rte}). The plug-in error from replacing $\delta_t$ by $\wh\delta_t$ in $\widetilde Y_{it}$ is
asymptotically negligible relative to the inverse Radon transform estimation.

For any $d = (d_1,d_2)\in\mathbb R^2$, define a Radon transform estimator \(\wh f_{C_i(d),S_i(d)}\) for the joint density of $(C_i(d),S_i(d))$ by replacing $v_i$ in (\ref{eq:est_cstar}) with
$v_i(d) \equiv \| (1,h_i')\|^{-1} (d_1 \tilde y_{i1} + d_2 \tilde y_{i2})$, where $\tilde y_{it} \equiv y_{it} - w_{it}'\delta_t$.

We establish uniform convergence of our estimators for the joint density of
$(\gamma_i,\beta_i)$ and the conditional density of $U_i$ given $(\gamma_i,\beta_i)$.
Recall these are all conditional on $Y_{i0}$, which is fixed at a constant value $y_0$ in the sample, and suppressed in notation.

We maintain the conditions for identification in \Cref{theorem:two} throughout, and do not reiterate them below for brevity.
In addition, we maintain the following assumptions.

\medskip

\begin{assumption}
    \label{assn:rc}
    Let $\mathcal T$ be a compact subset of the joint support of $(\gamma_i,\beta_i)$ such that
\begin{enumerate}[(i)]
    \item \label{assn:rc:jac} the Jacobian determinant $|J_i(d)|$ in \eqref{defn:Jacobian} is bounded away from zero for any fixed $d \in \mathbb R^2$ with $d_2\neq0$.
    \item \label{assn:rc:bound} $\sup_{(\gamma,b)\in\mathcal T}f_{\gamma_i,\beta_i}(\gamma,b)<\infty$.
    \item \label{assn:rc:positive}
    there exists a constant $M_b>0$ such that $\inf_{(\gamma,b)\in\mathcal T} f_{\gamma_i,\beta_i}(\gamma,b)\ge M_b$.
\end{enumerate}
\end{assumption}

\medskip

For a fixed $d$, the Jacobian $|J_i(d)|=|\delta_{1}d_2^{J_x+1}(d_1+d_2\gamma)^{J_x}|$ vanishes at $\gamma=-d_1/d_2$.
\Cref{assn:rc}\eqref{assn:rc:jac} excludes this value from the compact set $\mathcal T$.
\Cref{assn:rc}\eqref{assn:rc:positive} bounds $f_{\gamma_i,\beta_i}$ away from zero on $\mathcal T$, so that $\wh f_{\gamma_i,\beta_i}$ stays away from zero asymptotically.

\medskip

\begin{assumption}
\label{assn:f_cs}
    For any $d \in \mathbb R^2$ with $d_2\neq0$, the joint density $f_{C_i(d),S_i(d)}$ satisfies
\begin{enumerate}[(i)]
    \item \label{assn:f_cs:smooth} $f_{C_i(d),S_i(d)} \in \mathcal W^m(\mathbb R^l; B)$ for some $B > 0$ and smoothness order $m > l/2$, and
    \item \label{assn:f_cs:tail} there exist constants $M_c>0$ and $\kappa>l-1$ such that $f_{C_i(d),S_i(d)}(z)\le M_c\,(1+\|z\|)^{-\kappa}$ for all $z\in\mathbb R^l$.
\end{enumerate}
\end{assumption}

\medskip

The Sobolev smoothness in \Cref{assn:f_cs}\eqref{assn:f_cs:smooth} controls the bias of the Radon transform
estimator, and $m>l/2$ ensures that $f_{C_i(d),S_i(d)}$ is bounded and continuous.
This coincides with the condition of \cite{bissantz2014confidence} for the uniform bound on the bias.
The decay exponent $\kappa>l-1$
controls the variance of the Radon transform estimator.
\citet{holzmann2020} impose a similar polynomial tail bound that determines the optimal rate of their Fourier-based estimator.
\citet{goldenshluger2021deconvolution} impose a comparable tail bound on the target density.

\medskip

\begin{assumption}
\label{assn:f_Q}\leavevmode
\begin{enumerate}[(i)]
    \item \label{assn:f_Q:bound} The density $f_{Q_i}$ is uniformly bounded.
    \item \label{assn:f_Q:smooth} For some even integer $\zeta>(l-1)^2/(2m-l+1-2\varsigma)$ with some constant $\varsigma\in(0,m-l/2)$, the partial derivatives of $f_{Q_i}$ up to order $\zeta$ exist and are bounded.
    \item \label{assn:f_Q:positive} There exist constants $M_q>0$ and $k >1$ such that $f_{Q_i}(q)\ge M_q\,q_1^{\,k}$ for any $q\in\mathbb{S}_+^{l-1}$.
    \item \label{assn:f_Q:bw} The smoothing parameter of $\wh f_{Q_i}$ satisfies $\tau\asymp(\ln n/n)^{1/(2\zeta+l-1)}$.
  \item \label{assn:f_Q:sphK} The spherical kernel $\wt K$ is bounded, H\"older continuous, and of order at least $\zeta$.
\end{enumerate}
\end{assumption}

\medskip

\Cref{assn:f_Q}\eqref{assn:f_Q:bound} corresponds to the upper boundedness condition in Assumption~3 of \citet{hoderlein2010analyzing}; \eqref{assn:f_Q:positive} adapts the lower bound in their Assumption~4 to the trimmed setting, which yields \(M_{q,n}\equiv\inf_{q\in\mathbb S(\underline q_n)}f_{Q_i}(q)\ge M_q\underline q_n^k\).
The smoothness condition~\eqref{assn:f_Q:smooth}, together with the spherical kernel condition~\eqref{assn:f_Q:sphK}, gives the uniform rate of the spherical kernel estimator $\wh f_{Q_i}$
under \(\tau\asymp(\ln n/n)^{1/(2\zeta+l-1)}\); see \citet[Appendix~B]{hoderlein2010analyzing}.

The next assumption is on the kernel function $K_\nu(\cdot)=\nu^{-l}K(\cdot/\nu)$ in Footnote~\ref{note:defn_K}.

\medskip

\begin{assumption}
\label{assn:K_nu}\leavevmode
       \begin{enumerate}[(i)]
     \item \label{assn:K_nu:bound} The kernel $K$ is uniformly bounded.
           \item \label{assn:K_nu:lip} There exists a constant $M_k > 0$ such that, for all $z,z' \in \mathbb{R}$,
               $|K(z)- K(z')| \leq M_k\|z- z'\|$.
\item \label{assn:K_nu:bw} The smoothing parameter $\nu$ is of the order $\nu\asymp(\underline q_n^{1-k}\ln n/n)^{1/(2m+l-1-2\varsigma)}$.
                \item \label{assn:K_nu:order} The order parameter $r$ of the kernel $K_\nu$ satisfies $r\ge m-l/2$.
\end{enumerate}
\end{assumption}

\medskip

\Cref{assn:K_nu}\eqref{assn:K_nu:bound} and \eqref{assn:K_nu:lip} require the kernel $K$ to be bounded and Lipschitz continuous. They are standard for uniform convergence of kernel density estimators, as in \cite{li2007nonparametric} and \cite{masten2018random}.
\Cref{assn:K_nu}\eqref{assn:K_nu:bw} and \eqref{assn:K_nu:order} specify the rate of the bandwidth \(\nu\) and the order of the kernel.
Following Lemma~1 of \cite{bissantz2014confidence},
the order $r\ge m-l/2$ ensures the kernel does not limit the bias rate of the intermediate Radon transform estimator $f_{C_i(d),S_i(d)}$, which is then determined by the smoothness $m$.
\medskip

\begin{assumption}
\label{assn:f_u}
    For any $(\gamma,b)$ on the joint support of $(\gamma_i,\beta_i)$,
    \begin{enumerate}[(i)]
    \item \label{assn:f_u:tail} there exists $p>1$ such that $\sup_{(\gamma,b)\in\mathcal T} E(|U_{ij}|^{p}\mid\gamma_i=\gamma,\beta_i=b)<\infty$ for $j=1,2$.
    \item \label{assn:f_u:bound} $\max \{\int\sup_{u_2 \in \mathbb R}f_{U_i\mid\gamma_i=\gamma,\beta_i=b}(u_1,u_2)\,du_1, \int\sup_{u_1\in \mathbb R}f_{U_{i}\mid\gamma_i=\gamma,\beta_i=b}(u_1,u_2)\,du_2\}<\infty$.
    \item \label{assn:f_u:smooth} The conditional density $f_{U_{i}\mid\gamma_i=\gamma,\beta_i=b}$ is continuously differentiable, and both
    \(\int\sup_{u_1} (1+|u_2|)\,|\partial_{u_1} f_{U_{i}\mid\gamma,b}(u_1,u_2)|\,du_2\) and \(\int\sup_{u_2}(1+|u_1|)\,|\partial_{u_2} f_{U_{i}\mid\gamma,b}(u_1,u_2)|\,du_1\) are finite.
    \item \label{assn:f_u:char} There exist constants $M_{\phi}>0$ and $\rho>2$ such that $|\phi_{U_{i}\mid\gamma_i=\gamma,\beta_i=b}(a)|\le M_{\phi}(1+\|a\|)^{-\rho}$ for all $a\in\mathbb R^2$, where $\phi_{U_i\vert \gamma_i,\beta_i}(\cdot)$ denotes the characteristic function.
\end{enumerate}
\end{assumption}

\medskip

\Cref{assn:f_u}\eqref{assn:f_u:tail} and~\eqref{assn:f_u:char} bound the conditional moments of the random intercepts and the tail of their characteristic function.
Conditions (\ref{assn:f_u:tail}), (\ref{assn:f_u:bound}), and (\ref{assn:f_u:char}) are instrumental for dealing with the trimming devices $\mathcal C_n $, $\mathcal U_n$, and $\mathcal A_n $ in asymptotics.
The smoothness condition~(\ref{assn:f_u:smooth}) adapts Assumption~E1.3 of \citet{masten2018random} to non-compact support, yielding a bound on $\|\nabla_a f_{\wt U_i(a)\mid\gamma,b}\|$ that is analogous to his Lemma~S3 under compact support.

\Cref{assn:f_u}\eqref{assn:f_u:char} is the ordinary-smooth condition of the deconvolution literature \citep{fan1991optimal,meister2009deconvolution}. Since $(1+\|a\|)^{-\rho}$ is integrable on $\mathbb R^2$ for $\rho>2$, the conditional characteristic function is absolutely integrable, which is Assumption~E1.4 of \citet{masten2018random} for the Fourier inversion. The decay helps to account for the bias from truncating the inversion on $\mathcal A_n$.

Define $r_{1n}\equiv(\underline q_n^{1-k}\ln n/n)^{\frac{m-l/2-\varsigma}{2m+l-1-2\varsigma}}$, $r_{2n}\equiv(\ln n/n)^{\frac{\zeta}{2\zeta+l-1}}$, and let $$r_n\equiv r_{1n}+\underline q_n^{1/2} +\underline q_n^{1-k}r_{2n}r_{1n}^{-\frac{2(l-1)}{2m-l-2\varsigma}}.$$
\begin{assumption}\label{assn:tuning}
$\underline q_n\to0$,
$\bar c_n \asymp r_n^{-1/(p+1)}$, and
$$ \qquad \underline q_n^{-1}(\ln n/n)^{\min\left\{\frac{2m+l-2-2\varsigma}{k(4m+2l-3-4\varsigma)+1},\ \frac{\zeta}{k(2\zeta+l-1)}, \ \frac{\zeta(2m-l+1-2\varsigma)-(l-1)^2}{2(k-1)(m-\varsigma)(2\zeta+l-1)}\right\}}\to 0.$$
\end{assumption}

The sequence $r_n$ characterizes the uniform convergence rate of an intermediate density estimator $\widehat f_{C_i(d),S_i(d)}$ for a fixed index $d$ (\Cref{lemma:rte}).
\Cref{assn:tuning} restricts the trimming sequence $\underline q_n$ and the truncation sequence $\bar c_n$, and ensures $r_n \to 0$.

The next proposition establishes uniform convergence of our estimator for the density of $(\gamma_i,\beta_i)$.
Recall that $l = 2J_x+2$ denotes the dimension of $(1,H_i')'$ (with $W_{i1}\in\mathbb R$).

\medskip

\begin{proposition}\label{prop:margin}
Fix $d\in\mathbb R^2$ with $d_2\neq 0$, and suppose Assumptions \ref{assn:rc}\eqref{assn:rc:jac}--\eqref{assn:rc:bound}, \ref{assn:f_cs}, \ref{assn:f_Q}, \ref{assn:K_nu}, \ref{assn:f_u}\eqref{assn:f_u:tail} and~\ref{assn:tuning} hold.
Then, for any moment order $p$ in \Cref{assn:f_u}\eqref{assn:f_u:tail},
$$\sup_{(\gamma,b) \in \mathcal T} \left|\wh f_{\gamma_i,\beta_i}(\gamma,b) - f_{\gamma_i,\beta_i}(\gamma,b)\right| = O_p\!\left(r_n^{p/(p+1)}\right).$$
\end{proposition}

\begin{remark}
The adjustment factor $p/(p+1)$ is due to the accommodation of non-compact support of random intercepts $U_i$. The convergence rate is $r_n$ under compact support, or arbitrarily close to $r_n$ if the conditional moments are finite for all orders $p> 1$ (e.g., a Gaussian, a Gaussian mixture, or a logistic distribution).
\end{remark}

\begin{remark}
\citet{hoderlein2010analyzing} estimate the joint density of random coefficients in a cross-sectional linear model using a Radon transform estimator, including a trimmed version for models with an intercept.
In our setting, the Radon transform estimator is applied to the intermediate reduced-form density \(f_{C_i(d),S_i(d)}\).
\Cref{prop:margin} provides a sup-norm rate required for the conditional density of $U_i$ in \Cref{theorem:u}.
\end{remark}

We next establish uniform convergence of our estimator for the
joint density of $U_i \equiv (U_{i1},U_{i2})'$ conditional on
$(\gamma_i, \beta_i)$, which we recover by Fourier inversion of
the characteristic function estimator
$\wh \phi_{U_{i1},U_{i2} \mid \gamma_i,\beta_i}$.
The result is obtained under Assumptions S.1 and S.2 in Online Appendix S2.
Assumption S.1 imposes tail and smoothness bounds on the density of $(C_i(\bar d),S_i(\bar d))$ uniformly over $a \in \mathcal A_n$ almost everywhere. 
Assumptions S.2 imposes restrictions on the tuning sequences $\bar \vartheta_n$ and $\underline \vartheta_n$.

\medskip

\begin{theorem}\label{theorem:u}
Suppose Assumptions~\ref{assn:rc}, \ref{assn:f_cs}, \ref{assn:f_Q}, \ref{assn:K_nu}, \ref{assn:f_u}, \ref{assn:tuning}, S.1, S.2 hold, 
and that $E(|\wt Y_{i1}|)<\infty$ and $E(|\wt Y_{i2}|)<\infty$. Then, for any $(\gamma,b)$ on the support of $(\gamma_i,\beta_i)$,
    \begin{align*}
        \sup_{(u_1,u_2)\in \mathbb{R}^2}\left|\wh f_{U_{i1},U_{i2}|\gamma_i = \gamma, \beta_i=b}(u_1,u_2) -  f_{U_{i1},U_{i2}|\gamma_i = \gamma, \beta_i=b}(u_1,u_2)\right| = o_p(1).
    \end{align*}
\end{theorem}

\begin{remark}
\citet[Section~C.4]{masten2018random} proves sup-norm consistency of this Fourier inversion under compact support;
\Cref{theorem:u} permits non-compact support.
The proof bounds the error of the Fourier inversion of $\wh\phi$ over $\mathbb R^2$ by decomposing it into  $R_1+R_2+R_3$, where
\begin{equation*}
R_1=\frac{1}{(2\pi)^2}\int_{\mathcal A_n}|\wh\phi(a)-\phi(a)|\,da,\qquad
R_2+R_3=\frac{1}{(2\pi)^2}\int_{\mathbb R^2\setminus\mathcal A_n}|\phi(a)|\,da,
\end{equation*}
where $R_2$ and $R_3$ decompose further over $[-\bar{\vartheta}_n,\bar{\vartheta}_n]^2\setminus\mathcal A_n$ and $\mathbb R^2\setminus[-\bar{\vartheta}_n,\bar{\vartheta}_n]^2$, respectively.
We derive the convergence rates of these components.
The estimation error $R_1$ is due to three sources: the inverse Radon transform estimation, the trimming over $\mathcal A_n$, and the truncation of the non-compact intercept support.
The truncation errors satisfy $R_2 + R_3 = O(\underline{\vartheta}_n+\bar{\vartheta}_n^{-(\rho-2)})$, which follows from the characteristic function decay in \Cref{assn:f_u}\eqref{assn:f_u:char}.
\Cref{theorem:u} then follows from Assumption~S.2.
\end{remark}

\section{Monte Carlo Simulations}\label{sec:simulations}

In this section, we evaluate the finite-sample performance of the multi-step estimators for the joint density of $(\gamma_i,\beta_i)$ and the conditional density of $U_i \equiv (U_{i1},U_{i2})'$ given $(\gamma_i,\beta_i)$. We focus on the case with $T=2$ and fix the initial condition at $y_0=1$ throughout this simulation study.

The data-generating process (DGP) follows \eqref{eq:seq-exo-w}, with $(\delta_1,\delta_2)=(1,1)$. The regressors are generated independently of $(\gamma_i,\beta_i,U_i)$ as follows:
\[
X_i \equiv (X_{i1},X_{i2})'\sim \mathcal{N}(0,\Sigma_x),
\qquad
\Sigma_x=
\begin{pmatrix}
2 & 1 \\
1 & 2
\end{pmatrix},
\]
and
\(
W_{i2}=0.5\,W_{i1}+0.3\,W^2_{i1}+\varepsilon_{w,i},
\)
where $W_{i1}\sim\mathcal{N}(0,2.5)$ and $\varepsilon_{w,i}\sim\mathcal{N}(0,2.5)$.
The autoregressive coefficient is generated as $\gamma_i \sim \text{Beta}(6,3).$
The marginal distributions of $\beta_i$ and $U_i$ are given by
\begin{equation}\label{eq:marginals}
\beta_i \sim \mathcal{N}(\mu_\beta,\Sigma_\beta),
\qquad
U_i \sim \mathcal{N}(0,\Sigma_u),
\end{equation}
where
\[
\mu_\beta=(1,1)', \qquad
\Sigma_\beta=
\begin{pmatrix}
0.5 & 0.25 \\
0.25 & 0.5
\end{pmatrix},
\qquad
\Sigma_u=
\begin{pmatrix}
0.3 & 0.15 \\
0.15 & 0.3
\end{pmatrix}.
\]
\noindent We consider the following four designs.
\paragraph{Baseline.} The random coefficients and intercepts $(\gamma_i, \beta_i, U_i)$ are mutually independent, with marginals as specified above.
\paragraph{Correlated.} This design introduces dependence among $(\gamma_i, \beta_i, U_i)$ while preserving the conditional mean restriction in Assumption~\ref{assn:seq-exo}.
While $\gamma_i$ has the same Beta marginal as in the baseline design, $\beta_i$ and $U_i$ are constructed via
$$
    \beta_i = \mu_\beta + \exp(\lambda_\beta \gamma_i)\, L_\beta\, \eta_i^\beta, \qquad U_i = \exp(\lambda_u \gamma_i)\, L_u\, \eta_i^u,
$$
where $L_\beta$ and $L_u$ are the lower-triangular matrices from the Cholesky decompositions of $\Sigma_\beta$ and $\Sigma_u$, respectively.
The shocks $\eta_i = (\eta_i^\beta, \eta_i^u)'$ are drawn
independently of $\gamma_i$ from $\mathcal{N}(0, \Sigma_{\mathrm{corr}})$, and the $4$--by--$4$ correlation matrix $\Sigma_{\mathrm{corr}}$ has identity diagonal blocks and cross-block matrix
\[
\Sigma^{\beta,u} = \begin{pmatrix} 0.25 & 0.10 \\ 0.10 & 0.25 \end{pmatrix}.\]
We set $\lambda_\beta=\lambda_u=0.3$.
Under this design, conditional on $\gamma_i$, the vectors $\beta_i$ and $U_i$ have the same means as in the baseline design but covariance matrices that scale with $\gamma_i$:
$$\beta_i \mid \gamma_i \sim \mathcal{N}\left(\mu_\beta,\,
  \exp(2\lambda_\beta\gamma_i)\, \Sigma_\beta\right), \qquad
    U_i \mid \gamma_i \sim \mathcal{N}\left(0,\,
        \exp(2\lambda_u\gamma_i)\, \Sigma_u\right).$$

\paragraph{Scale-dependence.}
This design is simpler than the preceding correlated design in that there is no direct correlation between $\beta_i$ and $U_i$; their dependence is solely driven by $\gamma_i$ through the scale:
\[
\beta_i \mid \gamma_i \sim \mathcal{N}\left(\mu_\beta,\, v(\gamma_i)\,\Sigma_\beta\right),
\qquad
U_i \mid \gamma_i \sim \mathcal{N}\left(0,\, w(\gamma_i)\,\Sigma_u\right),
\]
where
\(
v(\gamma)=\exp\left[2\lambda_\beta(\gamma-\mathbb{E}[\gamma])\right]/Z(\lambda_\beta)\) and
\(w(\gamma)=\exp\left[2\lambda_u(\gamma-\mathbb{E}[\gamma])\right]/Z(\lambda_u),
\)
with $Z(\cdot)$ being a normalization constant such that $\mathbb{E}[v(\gamma_i)]=\mathbb{E}[w(\gamma_i)]=1$.

\paragraph{Bimodal.} The random coefficients $\gamma_i$ and $\beta_i$ are independent with the same marginal distributions as in the baseline design, but $U_i$ follows a symmetric Gaussian mixture:
\[
U_i \sim \tfrac12\,\mathcal{N}(\mu_u,\Sigma_u)
       + \tfrac12\,\mathcal{N}(-\mu_u,\Sigma_u),
\qquad \mu_u=(0.7,-0.7)'.
\]
\citet{hoderlein2010analyzing} and \citet{masten2018random} use similar Gaussian-mixture designs to investigate estimator performance for multi-modal densities.
For each design, we vary the sample sizes $n \in \{500, 2000, 8000\}$ and generate $N_{\mathrm{rep}} = 500$ Monte Carlo replications.
The estimators for the joint density $\wh f_{\gamma_i,\beta_i}$ and for the conditional density $ \wh f_{U_{i1},U_{i2} \mid \gamma_i=\gamma,\beta_i=b}$ are defined as in Sections~\ref{subsec:est_RC} and~\ref{subsec:est_U}.
Following \citet{hoderlein2010analyzing}, we select the bandwidths $(\tau, \nu)$ by minimizing the mean density-weighted integrated squared error.
Following \citet{masten2018random}, we obtain an optimal $\nu$ for a single $ a \in [-\bar{\vartheta}_n,\bar{\vartheta}_n]^2$ and then rescale it to the others. In the baseline design at $n=8000$, the average $\nu$ used over all grid points $a$ is $0.65$.
As the asymptotic theory prescribes, the cutoff $\bar{\vartheta}_n$ increases with the sample size, and we scale the bandwidths for the other sample sizes accordingly.

We first present the performance of $\wh f_{\gamma_i, \beta_i}$ through one-dimensional conditional densities. Table~\ref{tab:jgb_mise} reports the mean integrated squared error (MISE) of $\wh f_{\gamma_i \mid \beta_i}$, $\wh f_{\beta_{i1} \mid \gamma_i, \beta_{i2}}$, and $\wh f_{\beta_{i2} \mid \gamma_i, \beta_{i1}}$.
For every design and conditional density reported, the MISE declines with $n$, confirming the consistency of our estimator.

\begin{table}[H]
\begin{center}
\caption{Conditional densities of the estimated joint density of $(\gamma_i, \beta_i)$}
\label{tab:jgb_mise}
\scalebox{1}{
\begin{tabular}{ll ccc}
\toprule
& & \multicolumn{3}{c}{MISE} \\
\cmidrule(lr){3-5}
 & $n$ & $\wh f_{\gamma_i \mid \beta_i}$ & $\wh f_{\beta_{i1} \mid \gamma_i, \beta_{i2}}$ & $\wh f_{\beta_{i2} \mid \gamma_i, \beta_{i1}}$ \\
\midrule
 \multirow{3}{*}{Baseline}
 & 500  & 0.2359 & 0.1054 & 0.0654 \\
 & 2000 & 0.2001 & 0.0790 & 0.0476 \\
 & 8000 & 0.1749 & 0.0655 & 0.0369 \\
\midrule
\multirow{3}{*}{Correlated}
 & 500  & 0.2699 & 0.0664 & 0.0414 \\
 & 2000 & 0.2372 & 0.0468 & 0.0306 \\
 & 8000 & 0.2052 & 0.0374 & 0.0227 \\
\midrule
\multirow{3}{*}{Scale}
 & 500  & 0.2417 & 0.1031 & 0.0607 \\
 & 2000 & 0.2046 & 0.0792 & 0.0459 \\
 & 8000 & 0.1747 & 0.0644 & 0.0355 \\
\midrule
\multirow{3}{*}{Bimodal}
 & 500  & 0.2434 & 0.1066 & 0.0639 \\
 & 2000 & 0.2096 & 0.0816 & 0.0483 \\
 & 8000 & 0.1860 & 0.0695 & 0.0395 \\
\bottomrule
\end{tabular}}\end{center}

\vspace{12pt}

{\noindent\footnotesize \textit{Notes:} The table reports, across $N_{\mathrm{rep}} = 500$ Monte Carlo replications, the mean integrated squared error (MISE) of the estimated joint density $\wh f_{\gamma_i, \beta_i}$, summarized through its three one-dimensional conditional densities; each varies one coordinate of $(\gamma_i, \beta_{i1}, \beta_{i2})$ with the other two held at their medians.}
\end{table}

Figures~\ref{fig:marginal_Q50} and~\ref{fig:heatmap_Q50} show the estimated conditional density of $U_i$ along each axis and as contour plots, and Table~\ref{tab:cond_u} reports its MISE and the bias, standard deviation (SD), and mean squared error (MSE) of the estimated interquartile range (IQR) of its two marginals.
For the three unimodal designs (Baseline, Correlated, and Scale), the MISE falls with $n$ and the Monte Carlo mean density converges toward the true density, with the bias small in large samples.
The conditional location is recovered accurately: along each axis the estimated density stays centered on the true density even in small samples.
\citet[Figure~2 and Table~2]{masten2018random} likewise finds the location well recovered even when the full shape is hard to estimate.

The IQR in Table~\ref{tab:cond_u} is estimated less precisely in small samples but approaches its true value as $n$ grows. For the three unimodal designs it tends to be over-estimated, possibly due to over-smoothing at small $n$ which blurs the estimate: it moves probability mass away from the peak into the sides, so the estimated density is too flat and too wide, and its IQR too large (Figure~\ref{fig:marginal_Q50}).
\citet[Table~2]{masten2018random} reports the same upward IQR bias for his unimodal designs, where the estimated density is more spread out than the true density.
The Bimodal design is the most demanding. The MISE declines slowly ($0.054$ to $0.035$). Along each axis the estimate for the marginal density does not register a clear bimodal pattern even at $n=8000$ (Figure~\ref{fig:marginal_Q50}), although the contour separates its two modes by $n=2000$ (Figure~\ref{fig:heatmap_Q50}). Relatedly, the spread and hence the IQR are \emph{under}-estimated. Again, this pattern is consistent with \citet{masten2018random}'s observation that the shape is harder to estimate for multi-modal densities.

\begin{table}[ht]
\begin{center}
\caption{Estimator for the conditional density of $U_i$}
\label{tab:cond_u}
\scalebox{1}{
\begin{tabular}{ll c ccc ccc}
\toprule
& & & \multicolumn{3}{c}{$\wh{\mathrm{IQR}}(U_{i1} \mid \gamma_i, \beta_i)$} & \multicolumn{3}{c}{$\wh{\mathrm{IQR}}(U_{i2} \mid  \gamma_i, \beta_i)$} \\
\cmidrule(lr){4-6} \cmidrule(lr){7-9}
 & $n$ & MISE & Bias & SD & MSE & Bias & SD & MSE \\
\midrule
 \multirow{3}{*}{Baseline}
 & 500  & 0.0841 & 0.4571 & 0.0306 & 0.2099 & 0.4353 & 0.0294 & 0.1903 \\
 & 2000 & 0.0277 & 0.2313 & 0.0171 & 0.0538 & 0.2157 & 0.0172 & 0.0468 \\
 & 8000 & 0.0053 & 0.1039 & 0.0099 & 0.0109 & 0.0903 & 0.0099 & 0.0082 \\
\midrule
\multirow{3}{*}{Correlated}
 & 500  & 0.0516 & 0.4449 & 0.0293 & 0.1988 & 0.4256 & 0.0297 & 0.1820 \\
 & 2000 & 0.0163 & 0.2117 & 0.0171 & 0.0451 & 0.1920 & 0.0180 & 0.0372 \\
 & 8000 & 0.0044 & 0.0602 & 0.0125 & 0.0038 & 0.0435 & 0.0137 & 0.0021 \\
\midrule
\multirow{3}{*}{Scale}
 & 500  & 0.0833 & 0.4552 & 0.0311 & 0.2082 & 0.4331 & 0.0292 & 0.1884 \\
 & 2000 & 0.0271 & 0.2297 & 0.0182 & 0.0531 & 0.2130 & 0.0179 & 0.0457 \\
 & 8000 & 0.0055 & 0.1024 & 0.0097 & 0.0106 & 0.0871 & 0.0107 & 0.0077 \\
\midrule
\multirow{3}{*}{Bimodal}
 & 500  & 0.0538 & $-$0.2158 & 0.0738 & 0.0520 & $-$0.2236 & 0.0775 & 0.0560 \\
 & 2000 & 0.0421 & $-$0.2136 & 0.0626 & 0.0496 & $-$0.2129 & 0.0734 & 0.0507 \\
 & 8000 & 0.0346 & $-$0.1057 & 0.0489 & 0.0136 & $-$0.1246 & 0.0594 & 0.0191 \\
\bottomrule
\end{tabular}}\end{center}

\vspace{12pt}

{\noindent\footnotesize \textit{Notes:} The table reports, across $N_{\mathrm{rep}} = 500$ Monte Carlo replications, the mean integrated squared error (MISE) of the estimated conditional density $\wh f_{U_{i1},U_{i2}\mid \gamma_i, \beta_i}$; and the bias, standard deviation (SD), and mean squared error (MSE) of the estimated interquartile range $\wh{\mathrm{IQR}}(U_{ij} \mid \gamma_i, \beta_i)$, $j=1,2$, of the two marginal distributions. All quantities are evaluated at the median value of $(\gamma_i, \beta_i)$.}
\end{table}

\begin{figure}[p]
\begin{center}
\scalebox{0.9}{
\includegraphics[width=\textwidth,height=0.92\textheight,keepaspectratio]{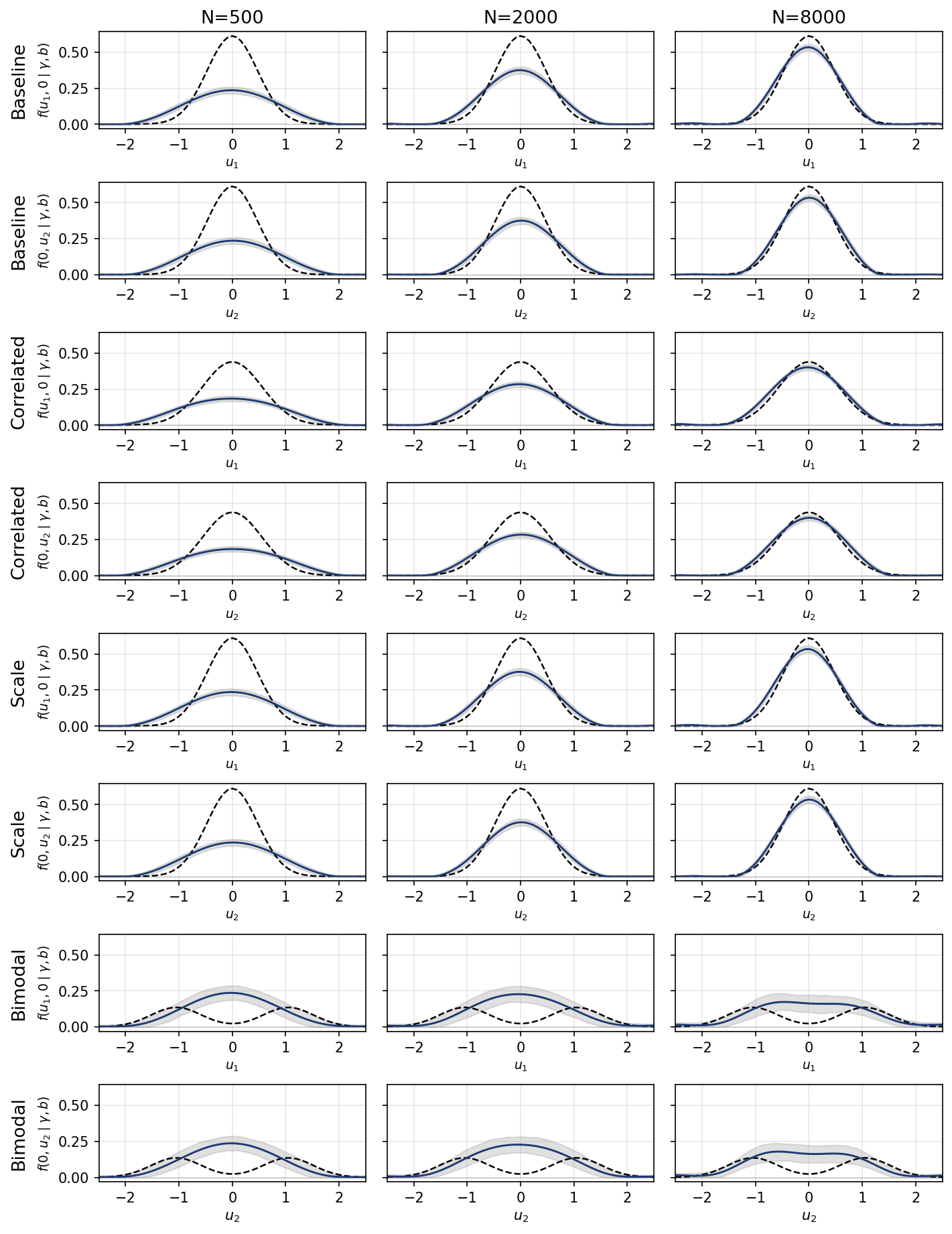}}
\caption{Estimated conditional density of $U_i$ along each axis}
\label{fig:marginal_Q50}
\end{center}

\vspace{12pt}

{\noindent\footnotesize  \textit{Notes:} $\wh f_{U_{i1}, U_{i2} \mid \gamma_i, \beta_i}$ along $u_1$ (at $u_2=0$) and along $u_2$ (at $u_1=0$), at the median value of $(\gamma_i,\beta_i)$. Each row shows these two curves for one design; columns are the sample sizes $n \in \{500, 2000, 8000\}$. Black dashed: the true density; navy solid: the Monte Carlo mean across $N_{\mathrm{rep}} = 500$ replications; shaded band: the 95\% pointwise Monte Carlo confidence band.}
\end{figure}

\begin{figure}[p]
\begin{center}
\scalebox{0.9}{
\includegraphics[width=\textwidth,height=0.92\textheight,keepaspectratio]{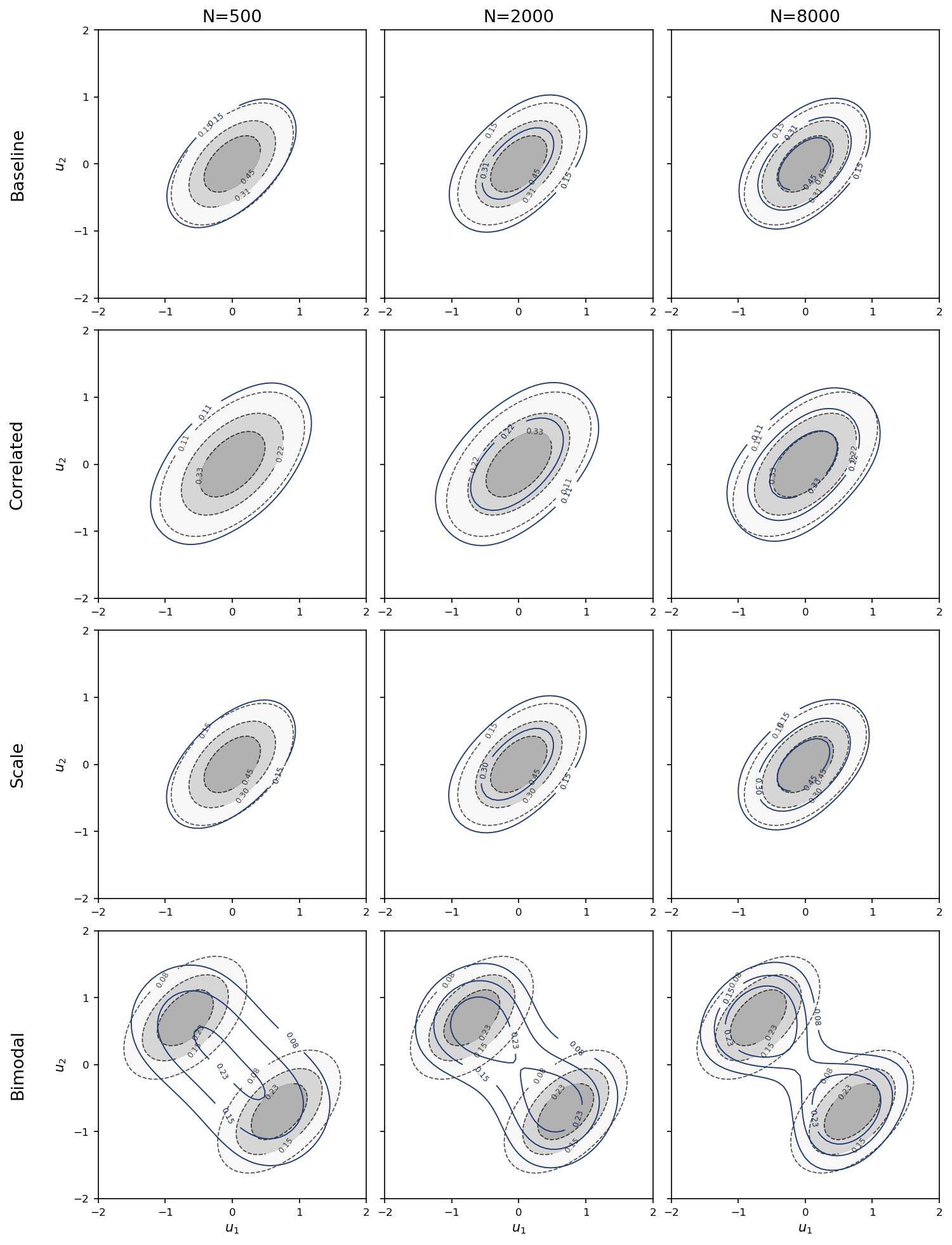}}
\caption{Contour plots of the estimated conditional density of $U_i$}
\label{fig:heatmap_Q50}
\end{center}

\vspace{12pt}

{\noindent\footnotesize  \textit{Notes:} $\wh f_{U_{i1}, U_{i2} \mid \gamma_i, \beta_i}$ for the four designs (rows) at sample sizes $n \in \{500, 2000, 8000\}$ (columns), conditioning on the median value of $(\gamma_i, \beta_i)$. Gray fill with black dashed contours: the true density; navy solid contours: the Monte Carlo mean across $N_{\mathrm{rep}} = 500$ replications. The contours of the true density enclose 25\%, 50\%, and 75\% of its probability mass, and the Monte Carlo contours use the same levels.}
\end{figure}

\FloatBarrier
\section{Conclusion}

We recover the entire joint distribution of the random coefficients and the time-varying errors. This is new relative to the literature on random-coefficient triangular and simultaneous equation models which left the distribution of time-varying errors unidentified. This is possible because a linear combination of the outcome history behaves as a single-equation random-coefficient regression in distinct regressors: inverting a Radon transform recovers a reduced-form density, and varying the combination and using deconvolution recover the distribution of the time-varying errors. We apply the analog principle to this constructive identification strategy, and propose a closed-form, multi-step estimator.

The identifying condition, a distributional form of strict exogeneity under which the covariates are independent of the coefficients and errors given the initial condition, is both the source of this reach and the main limitation of the approach: it permits arbitrary dependence among the coefficients and errors, but lets the covariates and the unobserved heterogeneity be related only through the initial condition.
Several questions remain for future research, including inference, data-driven bandwidth choices, and a continuously distributed initial condition.

\vspace{1cm}
\appendix
\renewcommand{\thesection}{\Alph{section}}
\noindent{\Large\textbf{Appendix}}

\section{Proofs in Section \ref{sec:asymp}}\label{sec:AppendixA}

We maintain the conditions for identification in \Cref{theorem:two} throughout, and do not reiterate them separately below for brevity.
Under the conditions of \Cref{theorem:two}, the conditional density $f_{V_i\mid Q_i}$, which equals the Radon transform of $f_{C_i(d),S_i(d)}$ over the directions on the support of $Q_i$, uniquely determines $f_{C_i(d),S_i(d)}$ within the class of distributions satisfying \Cref{assn:richSupp}.

The following lemma establishes uniform convergence of the Radon transform estimator $\wh f_{C_i(d), S_i(d)}$ for any $d \in \mathbb R^2$ with $d_2 \ne 0$.
\citet{hoderlein2010analyzing} proved unweighted $L^2$ consistency of the Radon transform estimator; we strengthen this to sup-norm convergence by combining the variance analysis as in \citet{li2007nonparametric} with the uniform bias bound of \cite{bissantz2014confidence}.

\medskip

\begin{lemma}[Sup-norm convergence of the RTE] \label{lemma:rte}
    Fix $d \in \mathbb R^2$ with $d_2 \ne 0$, and let $\mathcal S \subset \mathbb{R}^l$ be a compact subset of the joint support of $(C_i(d),S_i(d))$.
    Suppose Assumptions~\ref{assn:f_cs}, \ref{assn:f_Q}, \ref{assn:K_nu}, and \ref{assn:tuning} hold. Then,
    \[\sup_{(c,s) \in \mathcal S} \left| \wh f_{C_i(d),S_i(d)}(c,s) - f_{C_i(d),S_i(d)}(c,s) \right| = O_p\!\left(r_n\right),\]
    where $r_n \equiv r_{1n}+\underline q_n^{1/2}+\underline q_n^{\,1-k}r_{2n}\,
r_{1n}^{-2(l-1)/(2m-l-2\varsigma)}$.
\end{lemma}

\begin{proof}[Proof of \Cref{lemma:rte}]
Define an infeasible estimator
\begin{align*}
        \bar f_{C_i(d),S_i(d)}(c,s) = \frac{2}{n} \sum\nolimits_i \frac{\mathbf 1\{Q_{i1}\ge\underline q_n\}}{ f_{Q_i}(Q_i)}K_\nu(Q_i'(c,s')'-V_i(d)),
\end{align*}
and decompose \(\wh f_{C_i(d),S_i(d)}- f_{C_i(d),S_i(d)}\) as
\begin{equation*}
    \underbrace{\wh f_{C_i(d),S_i(d)}- \bar f_{C_i(d),S_i(d)}}_{I_1}
          + \underbrace{\bar f_{C_i(d),S_i(d)} - E(\bar f_{C_i(d),S_i(d)})}_{I_2}
           + \underbrace{E(\bar f_{C_i(d),S_i(d)})- f_{C_i(d),S_i(d)}}_{I_3}.
\end{equation*}
The Radon transform of $f:\mathbb R^l\rightarrow\mathbb R$ is $(Rf)(\varphi,\pi)\equiv \int_{\{z: \varphi'z=\pi\}}f(z)\,dz$, where $\pi \in \mathbb R$ and $\varphi \in \mathbb S^{l-1}$ with $\mathbb{S}^{l-1} \equiv \{z \in \mathbb{R}^l: \|z\|=1 \}$ being the unit sphere in $\mathbb R^l$.
A regularized inverse operator of the Radon transform, denoted by $A_\nu$, is a mapping from $\{g:\mathbb S_+^{l-1}\times\mathbb R\rightarrow\mathbb R\}$ to $\{f:\mathbb R^l\rightarrow\mathbb R\} $ defined as
\[
[A_\nu g](z)\equiv2\int_{\mathbb S_+^{l-1}}\int^\infty_{-\infty}K_{\nu}(\varphi'z-\pi)g(\varphi,\pi)\,d\pi\,d\mu(\varphi),
\]
where $K_\nu$ is a kernel function defined in Section \ref{subsec:est_RC} and $\mu$ denotes the Lebesgue measure on $\mathbb S_+^{l-1}$.
Define
\(
[A_{\nu,n}\, g](z)\equiv2\int_{\mathbb{S}(\underline q_n)}\int^\infty_{-\infty}K_{\nu}(\varphi'z-\pi)g(\varphi,\pi)\,d\pi\,d\mu(\varphi).
\)
We suppress $d$ in $C_i(d)$ and $S_i(d)$.
By construction,
\begin{align*}
    E(\bar f_{C_i,S_i}(c,s))
      & = \frac2n \sum\nolimits_i E\!\left[\frac{\mathbf 1\{Q_{i1}\ge\underline q_n\}}{f_{Q_i}(Q_i)}E\!\left(K_\nu(Q_i'(c,s')'-V_i)\mid Q_i\right) \right] \\
     &= \frac2n \sum\nolimits_i E\!\left[\frac{\mathbf 1\{Q_{i1}\ge\underline q_n\}}{f_{Q_i}(Q_i)} \int^{\infty}_{-\infty} K_\nu(Q_i'(c,s')'-v) (Rf_{C_i,S_i})(Q_i,v)\,dv \right] \\
     &=2\int_{\mathbb{S}(\underline q_n)} \left[\frac{1}{f_{Q_i}(q)} \int^{\infty}_{-\infty} K_\nu(q'(c,s')'-v) (Rf_{C_i,S_i})(q,v)\,dv\right]f_{Q_i}(q)\,d\mu(q)  \\
     &= [A_{\nu,n} (Rf_{C_i,S_i})](c,s),
\end{align*}
where the first equality uses the law of iterated expectations (with the outer expectation taken with respect to $Q_i$),
the second holds because $f_{V_i \mid Q_i}(v \mid q) = (Rf_{C_i, S_i})(q, v)$ which is implied by $(C_i,S_i)\perp Q_i$ under \Cref{assn:indep}, the third holds because the observations $Q_i$ are i.i.d.\ draws from the same density $f_{Q_i}$.

First, consider the bias term $I_3 = A_{\nu,n}(Rf_{C_i, S_i}) - f_{C_i, S_i}$, and decompose it as
\begin{equation*}
I_3 = \underbrace{\left(A_{\nu,n} - A_\nu\right)(Rf_{C_i,S_i})}_{I_{3,1}}
    + \underbrace{A_\nu (Rf_{C_i,S_i}) - f_{C_i,S_i}}_{I_{3,2}}.
\end{equation*}
Lemma S.1 bounds the two terms over $\mathbb R^l$, and since $\mathcal S\subset\mathbb R^l$,
\begin{equation*}
\sup_{(c,s)\in\mathcal S}|I_{3,1}|\le M_T\,B\,\underline q_n^{1/2},
\qquad
\sup_{(c,s)\in\mathcal S}|I_{3,2}|\le M_A\,B\,\nu^{m-l/2-\varsigma},
\end{equation*}
for some constants $M_T>0$ and $M_A>0$.
Then, by Lemma S.2,
\(
\sup_{(c,s)\in \mathcal S} |I_2| = O_p (\Xi_{q,n}^{1/2}\,\nu^{-l+1/2}\sqrt{\ln n / n}),
\)
where $\Xi_{q,n}\equiv\int_{\mathbb S(\underline q_n)}f^{-1}_{Q_i}(q)\,d\mu(q)$.
By \Cref{assn:f_Q}\eqref{assn:f_Q:positive},
\begin{align}\label{eq:xi_bound}
  \begin{split}
 \Xi_{q,n} &\le M_q^{-1}\int_{\mathbb S(\underline q_n)} q_1^{-k}\,d\mu(q)
    = M_q^{-1}\omega_{l-1}\int_{\underline q_n}^{1}q_1^{-k}(1-q_1^2)^{(l-3)/2}\,dq_1 \\
& \le M_q^{-1}\omega_{l-1}\int_{\underline q_n}^{1}q_1^{-k}\,dq_1
=\frac{\omega_{l-1}}{M_q\,(k-1)}\,(\underline q_n^{\,1-k}-1)
=O\left(\underline q_n^{\,1-k}\right),
  \end{split}
\end{align}
where $\omega_{l-1}$ is the surface area of $\mathbb S^{l-2}$, the first equality follows from Equation~(1.5.4) in \citet{daixu2013approximation}, the second inequality uses $(1-q_1^2)^{(l-3)/2}\le1$ on $[0,1]$ as $l>3$, and the last equality uses $k>1$.
Then,
\[
\sup_{(c,s)\in \mathcal S} |I_2| = O_p\!\left(\underline q_n^{(1-k)/2}\nu^{-l+1/2}\sqrt{\ln n/n}\right).
\]
Balance the rates of $I_2$ and $I_{3,2}$ by setting $\underline q_n^{(1-k)/2}\nu^{-l+1/2}\sqrt{\ln n/n} \asymp \nu^{m-l/2-\varsigma}$; this yields
$\nu \asymp (\underline q_n^{1-k} \ln n /n)^{1/(2m+l-1-2\varsigma)}$ (\Cref{assn:K_nu}\eqref{assn:K_nu:bw}).
Substituting back and applying the triangle inequality,
\(\sup_{(c,s)\in\mathcal S}\lvert I_2 + I_{3,2}\rvert = O_p(r_{1n}).\)
By Lemma S.3 and \eqref{eq:xi_bound},
$$\sup_{(c,s)\in\mathcal S}|I_1|=O_p(\nu^{-l+1}\Xi_{q,n}r_{2n})=O_p(\underline q_n^{\,1-k}r_{2n}\,
r_{1n}^{-2(l-1)/(2m-l-2\varsigma)}).$$

Combining these with the bound on the trimming bias $I_{3,1}$,
\begin{equation*}
  \sup_{(c,s)\in\mathcal S}\left|\wh f_{C_i,S_i}-f_{C_i,S_i}\right|=O_p\left(r_{1n}+\underline q_n^{1/2}+\underline q_n^{\,1-k}r_{2n}\,
r_{1n}^{-2(l-1)/(2m-l-2\varsigma)}\right).
\end{equation*}
\end{proof}

\medskip

\begin{proof}[Proof of \Cref{prop:margin}]
To simplify notation, we suppress $d$ in $C_i(d), S_i(d)$ and $J_i(d)$ throughout this proof.
Recall $\mathcal C_n\equiv[-\bar c_n,\bar c_n]$ with $\bar c_n \to \infty$ as $n \to \infty$.
The estimator $\wh f_{\gamma_i,\beta_i}(\gamma,b)$ integrates over $\mathcal C_n$, whereas $f_{\gamma_i,\beta_i}(\gamma,b)$ integrates over the support of $C_i$ conditional on $(\gamma,b)$.
By the triangle inequality,
\begin{align}\label{eq:prop3_decom}
    \begin{split}
      \sup_{(\gamma,b)\in\mathcal T}\left|\wh f_{\gamma_i,\beta_i}(\gamma,b)-f_{\gamma_i,\beta_i}(\gamma,b)\right|
& \le\underbrace{\sup_{(\gamma,b)\in\mathcal T}\int_{\mathcal C_n}\left|\wh f_{C_i,\gamma_i,\beta_i}(c,\gamma,b)-f_{C_i,\gamma_i,\beta_i}(c,\gamma,b)\right|dc}_{T_1}\\
&+\underbrace{\sup_{(\gamma,b)\in\mathcal T}\int_{|c|>\bar c_n}f_{C_i,\gamma_i,\beta_i}(c,\gamma,b)\,dc}_{T_2}.
    \end{split}
\end{align}

\emph{For $T_1$}, we have
\begin{align*}
T_1
&\le 2 \bar c_n\,\sup_{(c,\gamma,b) \in \mathcal C_n \times \mathcal T} \left|\wh f_{C_i,\gamma_i,\beta_i}(c,\gamma,b) - f_{C_i,\gamma_i,\beta_i}(c,\gamma,b) \right| \\
&=2 \bar c_n\,\sup_{(c,\gamma,b) \in \mathcal C_n \times \mathcal T} \left\{\left|\wh f_{C_i,S_i}(c,\tilde s)- f_{C_i,S_i}(c,\tilde s)\right|\,|J_i| \right\} \\
&\le 2 \bar c_n\,\sup_{(\gamma,b)\in\mathcal T} |J_i|\,\sup_{(c,\gamma,b) \in \mathcal C_n \times \mathcal T} \left|\wh f_{C_i,S_i}(c,\tilde s)- f_{C_i,S_i}(c,\tilde s)\right|,
\end{align*}
where $\tilde s \equiv (b_1(d_1+d_2\gamma), b_2d_2, d_2\delta_1\gamma )$,  the second step holds under \Cref{assn:rc}\eqref{assn:rc:jac}, and  $\sup_{(\gamma,b)\in\mathcal T}|J_i|<\infty$ because $|J_i|$ is continuous on the compact set $\mathcal T$. For any \((c,\gamma,b)\in\mathcal C_n \times \mathcal T\), \((c,\tilde s)\) belongs to the set \(\mathcal S_n \equiv \mathcal C_n \times \{\tilde s: (\gamma,b)\in\mathcal T\},\) where
the second component is compact because $\tilde s$ is continuous on $\mathcal T$.

\Cref{lemma:rte} establishes the uniform convergence rate of $\wh f_{C_i,S_i}(c,s)$ over a compact set $\mathcal S$; its proof applies to $\mathcal S_n$ because the set enters the proof (via Lemmas S.2 and S.3) only through the number of cubes that cover $\mathcal S_n$,
$N_n\asymp|\mathcal S_n|\,\ell_n^{-l},$
where $|\mathcal S_n|$ denotes the Lebesgue measure of $\mathcal S_n$ and $\ell_n$ is the side length of the cubes.
By Assumption~\ref{assn:tuning}, \(\bar c_n\asymp r_n^{-1/(p+1)}\) and the restriction on \(\underline q_n\) imply
\(\bar c_n=O(n^C)\) for some finite constant \(C>0\), so $N_n =o(n^{2l+C}(\ln n)^{-2l})$.
Thus, by an argument similar to that in the proof of
Lemma S.2, applying the Borel--Cantelli lemma
and a summability condition, the rate of \Cref{lemma:rte} holds uniformly over $\mathcal S_n$. Hence
\begin{align}\label{eq:prop3_bound}
  T_1=O_p\left(\bar c_n\,r_n\right).
\end{align}

\emph{For $T_2$}, write
\begin{align*}
  T_2 &=
   \sup_{(\gamma,b)\in\mathcal T} f_{\gamma_i,\beta_i}(\gamma,b) \,\int_{|c|>\bar c_n}f_{C_i \mid \gamma_i,\beta_i}(c,\gamma,b)\,dc  \\
  &\le  \bar M \sup_{(\gamma,b)\in\mathcal T} \Pr(|C_i|>\bar c_n \mid \gamma_i=\gamma,\beta_i=b)
  \le \bar M \,{\bar c_n^{\,-p}}{\sup_{(\gamma,b)\in\mathcal T}E\!\left(|C_i|^{p}\mid\gamma_i=\gamma,\beta_i=b\right)},
\end{align*}
where $\bar M\equiv\sup_{(\gamma,b)\in\mathcal T}f_{\gamma_i,\beta_i}(\gamma,b)<\infty$ by \Cref{assn:rc}\eqref{assn:rc:bound} and the last inequality follows from Markov's inequality with moment order $p$ of \Cref{assn:f_u}\eqref{assn:f_u:tail}.
Since $C_i \mid \gamma,b = (d_1\gamma+d_2\gamma^2)y_0 + (d_1+d_2\gamma)U_{i1} + d_2 U_{i2}$ and $p>1$, the $c_r$-inequality gives
\begin{equation}\label{eq:mean_C}
    E(|C_i|^{p}\mid\gamma,b) \le 3^{p-1}\Bigl(|(d_1\gamma+d_2\gamma^2)y_0|^{p} + |d_1+d_2\gamma|^{p}\,E(|U_{i1}|^{p}\mid\gamma,b) + |d_2|^{p}\,E(|U_{i2}|^{p}\mid\gamma,b)\Bigr),
\end{equation}
where $(d_1\gamma+d_2\gamma^2)y_0$ and $d_1+d_2\gamma$ are continuous in $\gamma$ and hence bounded on the compact set $\mathcal T$.
Together with \Cref{assn:f_u}\eqref{assn:f_u:tail}, the bound in \eqref{eq:mean_C} yields
$\sup_{(\gamma,b)\in\mathcal T}E(|C_i|^{p}\mid\gamma,b)<\infty$.
Then, we have
\begin{align}\label{eq:prop3_tail}
  T_2=O(\bar c_n^{\,-p}).
\end{align}
Combining \eqref{eq:prop3_bound} and \eqref{eq:prop3_tail},
$
\sup_{(\gamma,b)\in\mathcal T}|\wh f_{\gamma_i,\beta_i}-f_{\gamma_i,\beta_i}|
=O_p\left(\bar c_n r_n\right)+O(\bar c_n^{\,-p}).$
Setting $\bar c_n\asymp r_n^{-1/(p+1)}$ balances the two terms. Then
\begin{equation*}
\sup_{(\gamma,b)\in\mathcal T}|\wh f_{\gamma_i,\beta_i}-f_{\gamma_i,\beta_i}|=O_p\left(r_n^{\,p/(p+1)}\right),
\end{equation*}
where $r_n \to 0$ under \Cref{assn:tuning}.
\end{proof}

\medskip

\begin{proof}[Proof of \Cref{theorem:u}]
The proof is an adaptation of Theorem S2 in \cite{masten2018random}.
By the construction in Section~\ref{subsec:est_U},
\begin{align*}
&\wh f_{U_{i1},U_{i2}|\gamma_i = \gamma, \beta_i=b}(u_1,u_2) - f_{U_{i1},U_{i2}|\gamma_i = \gamma, \beta_i=b}(u_1,u_2)\\
&= \mathrm{Re}\!\left[\frac{1}{(2\pi)^2}\left(\int_{\mathcal{A}_n}\exp[-i(a_1u_1+a_2u_2)]\,\wh\phi(a)\,da - \int_{\mathbb R^2}\exp[-i(a_1u_1+a_2u_2)]\,\phi(a)\,da\right)\right],
\end{align*}
where $\phi(a) \equiv \phi_{U_{i1},U_{i2}|\gamma_i = \gamma, \beta_i=b}(a)$ is the conditional characteristic function of $U_i$ and $\wh\phi$ its estimator.
Using $|\mathrm{Re}(z)|\le|z|$ and applying the triangle inequality to complex-valued integrands, we have
\begin{align*}
&\left|\wh f_{U_{i1},U_{i2}|\gamma_i = \gamma, \beta_i=b}(u_1,u_2) - f_{U_{i1},U_{i2}|\gamma_i = \gamma, \beta_i=b}(u_1,u_2)\right| \\
&\le  \frac{1}{(2\pi)^2}\left|\int_{\mathcal{A}_n}\exp[-i(a_1u_1+a_2u_2)]\,\wh\phi(a)\,da - \int_{\mathbb R^2}\exp[-i(a_1u_1+a_2u_2)]\,\phi(a)\,da\right|\\
&\le  \underbrace{\frac{1}{(2\pi)^2}\int_{\mathcal{A}_n}|\wh\phi(a) - \phi(a)|\,da}_{R_1(\bar{\vartheta}_n,\underline{\vartheta}_n)}
+ \underbrace{\frac{1}{(2\pi)^2}\int_{[-\bar{\vartheta}_n,\bar{\vartheta}_n]^2\setminus\mathcal{A}_n}|\phi(a)|\,da}_{R_2(\underline{\vartheta}_n) }
+ \underbrace{\frac{1}{(2\pi)^2}\int_{\mathbb R^2\setminus[-\bar{\vartheta}_n,\bar{\vartheta}_n]^2}|\phi(a)|\,da}_{R_3(\bar{\vartheta}_n)},
\end{align*}
where the second inequality follows from $|\exp[-i(a_1u_1+a_2u_2)]| = 1$ for any $(a,u)$.

\textbf{\emph{Bound on $R_1$.}}
Recall that $\wh\phi(a) = \int_{\mathcal U_n}\exp(i\tilde u)\,\wh f_{\wt U_i(a)|\gamma_i =\gamma, \beta_i=b}(\tilde u)\,d\tilde u,$
while the true characteristic function $\phi(a)$ integrates $f_{\wt U_i(a)|\gamma_i =\gamma, \beta_i=b}(\tilde u)$ over $\mathbb R$.
 Let $C_i(a,\gamma) \equiv \wt U_i(a)+a_1\gamma y_0$, and $\mathcal U_n \equiv\{\tilde u \in \mathbb R: |\tilde u + a_1\gamma y_0| \leq \bar c_n\}$.
By the triangle inequality and $|\exp(i\tilde u)|=1$,
\begin{align*}
|\wh\phi(a) - \phi(a)|
&\le \int_{\mathcal U_n}\left|\wh f_{\wt U_i(a)|\gamma,b}(\tilde u)
- f_{\wt U_i(a)|\gamma,b}(\tilde u)\right|\,d\tilde u
+ \int_{\mathbb R \setminus \mathcal U_n}
f_{\wt U_i(a)|\gamma,b}(\tilde u)\,d\tilde u\\
& =\int_{\mathcal C_n}\big|\wh f_{C_i(a,\gamma)\mid\gamma,b}(c)
- f_{C_i(a,\gamma)\mid\gamma,b}(c)\big|\,dc
+ \int_{\mathbb R \setminus \mathcal C_n}
f_{C_i(a,\gamma)\mid\gamma,b}(c)\,dc\\
& \le 2\bar c_n \sup_{c\in\mathcal C_n} \big|\wh f_{C_i(a,\gamma)\mid\gamma,b}(c)
- f_{C_i(a,\gamma)\mid\gamma,b}(c)\big|
+ \Pr\!\left(|C_i(a,\gamma)|>\bar c_n \mid \gamma, b\right) \\
&\le 2\bar c_n \sup_{c\in\mathcal C_n} \big|\wh f_{C_i(a,\gamma)\mid\gamma,b}(c)
- f_{C_i(a,\gamma)\mid\gamma,b}(c)\big|
+\bar c_n^{\,-p}\,E\!\left(|C_i(a,\gamma)|^p \mid \gamma,b\right),
\end{align*}
where the equality follows from
the change of variables $c=\tilde u+a_1\gamma y_0$, which maps
$\mathcal U_n$ onto $\mathcal C_n$, the second inequality uses that
the Lebesgue measure of $\mathcal C_n$ is $2\bar c_n$, and the last inequality
applies Markov's inequality at moment order $p$ of
\Cref{assn:f_u}\eqref{assn:f_u:tail}.
By \eqref{eq:mean_C} in the proof of \Cref{prop:margin}, for $a\in\mathcal A_n$ and $p> 1$,
$$
E(|C_i(a,\gamma)|^{p}\mid\gamma,b) \le 3^{p-1}\Bigl(|a_1\gamma y_0|^{p} + |a_1|^{p}\,E(|U_{i1}|^{p}\mid\gamma,b) + |a_2|^{p}\,E(|U_{i2}|^{p}\mid\gamma,b)\Bigr)=O(\bar \vartheta_n^{p}),$$
where the equality uses $|a_1|, |a_2|\le\bar{\vartheta}_n$ on $\mathcal A_n$. Then,
\[
\sup_{a\in\mathcal A_n}|\wh\phi(a) - \phi(a)|  \leq 2\bar c_n\sup_{a\in\mathcal A_n}\sup_{c\in\mathcal C_n} \big|\wh f_{C_i(a,\gamma)\mid\gamma,b}(c)
- f_{C_i(a,\gamma)\mid\gamma,b}(c)\big|  + O(\bar c_n^{\,-p} \bar{\vartheta}_n^{p}).
\]
Using $|\mathcal A_n|\le 4\bar{\vartheta}_n^2$, Lemma A.2, and $\bar c_n\asymp r_n^{-1/(p+1)}$, we have
\begin{align*}
R_1
= O_p\!\left(\bar{\vartheta}_n^{2}\bar c_n\,\varepsilon_n^{*}\right)
+ O\!\left(\bar c_n^{\,-p}\bar{\vartheta}_n^{p+2}\right)= O_p\!\left(\bar{\vartheta}_n^{2}r_n^{-1/(p+1)}\varepsilon_n^{*}\right)
+ O\!\left(\bar{\vartheta}_n^{p+2}r_n^{\,p/(p+1)}\right),
\end{align*}
where $\varepsilon_n^{*}$ denotes the uniform convergence
rate of the estimator $\wh f_{C_i(a,\gamma) \mid \gamma, \beta}$ over
$\mathcal A_n$, as defined in Assumption S.2.

\textbf{\emph{Bound on $R_2$.}} The region $[-\bar{\vartheta}_n,\bar{\vartheta}_n]^2\setminus\mathcal{A}_n$ lies in $\{a:|a_1|<\underline{\vartheta}_n\}\cup\{a:|a_2|<\underline{\vartheta}_n\}$.
Since $\|a\|\ge|a_2|$, \Cref{assn:f_u}\eqref{assn:f_u:char} gives $|\phi(a)|\le M_\phi(1+\|a\|)^{-\rho}\le M_\phi(1+|a_2|)^{-\rho}$ for $\rho >2$.
Thus, integrating $|\phi(a)|$ over the set $\{a: |a_1|<\underline{\vartheta}_n\}$ yields
\begin{align*}
 \frac{1}{(2\pi)^2}\int_{|a_1|<\underline{\vartheta}_n}\!\int_{\mathbb R}|\phi(a)|\,da_2\,da_1
 \le \frac{2M_\phi\,\underline{\vartheta}_n}{(2\pi)^2}\int_{\mathbb R}(1+|a_2|)^{-\rho}\,da_2 = \frac{M_\phi\,\underline{\vartheta}_n}{\pi^2(\rho-1)} = O(\underline{\vartheta}_n),
\end{align*}
and the set $\{a:|a_2|<\underline{\vartheta}_n\}$ is bounded symmetrically using $\|a\|\ge|a_1|$. Hence
$R_2 = O(\underline{\vartheta}_n).$

\textbf{\emph{Bound on $R_3$.}} Since $\mathbb R^2\setminus[-\bar{\vartheta}_n,\bar{\vartheta}_n]^2\subseteq\{a:\|a\|>\bar{\vartheta}_n\}$, under \Cref{assn:f_u}\eqref{assn:f_u:char},
\begin{align*}
R_3 &\le \frac{M_\phi}{(2\pi)^2}\int_{\|a\|>\bar{\vartheta}_n}(1+\|a\|)^{-\rho}\,da
 = \frac{M_\phi}{(2\pi)^2}\int_{0}^{2\pi}\!\!\int_{\bar{\vartheta}_n}^{\infty}(1+t)^{-\rho}\,t\,dt\,d\theta
\le \frac{M_\phi}{2\pi}\int_{\bar{\vartheta}_n}^{\infty}t^{1-\rho}\,dt \\
& = \frac{M_\phi\,\bar{\vartheta}_n^{\,-(\rho-2)}}{2\pi(\rho-2)}
= O\!\left(\bar{\vartheta}_n^{-(\rho-2)}\right),
\end{align*}
where the first equality uses polar coordinates $a=(t\cos\theta,t\sin\theta)$, the second inequality uses $\int_{0}^{2\pi}d\theta=2\pi$ and $(1+t)^{-\rho}\,t\le t^{1-\rho}$ for $t \ge \bar \vartheta_n>0$,
and the second equality holds for $\rho > 2$.
Because the bounds on $R_1$, $R_2$, and $R_3$ do not depend on $(u_1,u_2)$,
\begin{align*}
&\sup_{(u_1,u_2)\in\mathbb R^2}\left|\wh f_{U_{i1},U_{i2}|\gamma_i =\gamma,\beta_i=b}(u_1,u_2) - f_{U_{i1},U_{i2}|\gamma_i =\gamma,\beta_i=b}(u_1,u_2)\right| \\
&=O_p\!\left(\bar{\vartheta}_n^{2}r_n^{-1/(p+1)}\varepsilon_n^{*}\right)
+ O\!\left(\bar{\vartheta}_n^{p+2}r_n^{\,p/(p+1)}\right)
+ O\!\left(\underline{\vartheta}_n\right) + O\!\left(\bar{\vartheta}_n^{-(\rho-2)}\right)=o_p(1),
\end{align*}
where $\rho>2$, the first two terms vanish under Assumption S.2, and the last two vanish because $\underline{\vartheta}_n\to0$ and $\bar{\vartheta}_n\to\infty$.
\end{proof}

\section{Extension: Long Horizons with \texorpdfstring{$T\ge 3$}{Tge3}}\label{sec:T3}

Our identification extends to panels with longer horizons ($T\ge 3$). To fix ideas, consider a model with only strictly exogenous covariates $X_{it}$. Suppose
\[
Y_{it}=\gamma_iY_{it-1}+X_{it}'\beta_{it}+U_{it},
\qquad t=1,\ldots,T,
\]
and let $H_i\equiv(X_{i1}',\ldots,X_{iT}')'$. Recursive substitution gives
\[
\sum_{t=1}^Td_tY_{it}
=
\underset{C_i(d)}{\underbrace{\Bigl(\sum_{t=1}^Td_t\gamma_i^t\Bigr)Y_{i0}
+
\sum_{s=1}^Tp_s(\gamma_i;d)U_{is}}}
+
\sum_{s=1}^TX_{is}'\,\underset{S_{is}(d)}{\underbrace{p_s(\gamma_i;d)\,\beta_{is}}},
\]
where $p_s(\gamma;d)\equiv\sum_{t=s}^Td_t\gamma^{t-s}$. This reduces to \eqref{eq:ty} without the $W_{i1}$ term when $T=2$.

\subsection{Time-invariant random slopes}

Suppose that the random coefficients for strictly exogenous covariates are time-invariant ($\beta_{it}=\beta_i\in\mathbb R^{J_x}$ for all $t=1,\ldots,T$). Then
\[
S_{i,T-1}(d)
=
(d_{T-1}+d_T\gamma_i)\beta_i,
\qquad
S_{iT}(d)
=
d_T\beta_i.
\]
In addition to the analogs of \Cref{assn:indep} and \ref{assn:richSupp} maintained throughout, we assume $\Pr\{\beta_{i,k} = 0 \mid Y_{i0} \} = 0$ for some $k$.
For any $d$ with $d_T\neq0$ and any $k$ with $\Pr\{\beta_{i,k}=0\mid Y_{i0}\}=0$, the map $(\gamma_i,\beta_i)\mapsto(S_{i,T-1,k}(d),S_{iT}(d))$ is invertible with Jacobian determinant $|d_T^{J_x+1}\beta_{i,k}|$.
The joint density of $(C_i(d),\gamma_i,\beta_i)$ given $Y_{i0}$ then follows from a change of variables similar to (\ref{eq:iden_CRC}), with the remaining slopes providing over-identifying restrictions.

It remains to recover the conditional distribution of time-varying errors $U_{it}$. For any $a \equiv (a_1,\ldots, a_T)$ and realized coefficient $\gamma_i = \gamma$, define
\[
d_s(a,\gamma) =a_s-\gamma a_{s+1} \text{ for }s=1,\ldots,T, \text{ with }
a_{T+1}\equiv0,
\]
so that $p_s(\gamma;d(a,\gamma))=a_s$. The characteristic function of
$(U_{i1},\ldots,U_{iT})$ conditional on
$(\gamma_i,\beta_i,Y_{i0})$ is then recovered exactly as in \eqref{eq:iden_tU}; continuity handles the boundary case $a_T=0$. The random coefficients and intercepts are thus identified under moment determinacy and support conditions as stated in the main text.
A sufficient condition for the moment determinacy condition at $T\ge 3$ is that, conditional on $Y_{i0}$, $\gamma_i$ has bounded support while the components of $(\beta_i,U_i)$ have sub-Gaussian (or merely sub-exponential) tails.

\subsection{Time-varying random slopes}

If $\beta_{it}$ are time-varying, then for any fixed $d$, the mapping from $(\gamma_i,\beta_{i1},\ldots,\beta_{iT})\in\mathbb R^{TJ_x+1}$ to $(S_{i1}(d),\ldots,S_{iT}(d))\in\mathbb R^{TJ_x}$ is not invertible.
Nevertheless, varying $d$ jointly with the direction in which $(S_{i1}(d),\ldots,S_{iT}(d))$ is evaluated helps to identify the full joint distribution of $(\gamma_i,\beta_{i1},\ldots,\beta_{iT})$.
We maintain the following conditions, which are analogs of \Cref{assn:indep} and \ref{assn:richSupp}. \medskip

\begin{assumption}\label{ass:B1}\leavevmode
\begin{enumerate}[(i)]
    \item\label{ass:B1:i} $(\gamma_i,\beta_i,U_i)\perp H_i\mid Y_{i0}$, and $(\gamma_i,\beta_i,U_i)\mid Y_{i0}$ admits a Lebesgue density almost
    surely.
    \item\label{ass:B1:ii} The support of $H_i\mid Y_{i0}$ contains an open ball in $\mathbb{R}^{TJ_x}$ almost surely;
    and for every $d\in\mathbb{R}^T$, the conditional distribution of $(C_i(d),S_{i1}(d),\ldots,S_{iT}(d))$ given $Y_{i0}$ is uniquely determined by its moments and has finite absolute moments of all orders almost surely.
\end{enumerate}
\end{assumption}

\medskip

Under Assumption~\ref{ass:B1}, the joint distribution of $(C_i(d),S_{i1}(d),\ldots,S_{iT}(d))$ given $Y_{i0}$ is identified from that of $\bigl(\sum_{t=1}^Td_tY_{it},\,H_i'\bigr)$ given $Y_{i0}$ for every $d\in\mathbb{R}^T$.
This is the general-$T$ analog of \Cref{lm:one} and follows from Lemma 2 of \citet{masten2018random} by the same argument, with $H_i\in\mathbb{R}^{TJ_x}$ playing the role of the regressors. \medskip

\begin{proposition}\label{prop:B1}
Suppose Assumption~\ref{ass:B1} holds. Then, for any fixed $T\ge2$ and
$J_x\ge1$, the joint distribution of
$(\gamma_i,\beta_{i1},\ldots,\beta_{iT})$ given $Y_{i0}$ is identified.
\end{proposition}

\begin{proof}[Proof of \Cref{prop:B1}]
The proof is conditional on $Y_{i0}$, which we suppress in notation. The target is the joint characteristic function (c.f.) of $(\gamma_i,\beta_{i1}',\ldots,\beta_{iT}')$:
    \[  \Phi(g,t) \equiv E\Bigl[\exp\Bigl\{i\Bigl(g\gamma_i +\sum\nolimits_{s=1}^Tt_s'\beta_{is}\Bigr)\Bigr\}\Bigr],
        \qquad g\in\mathbb{R},\qquad t=(t_1',\ldots,t_T')'\in\mathbb{R}^{TJ_x}. \]
Since $\Phi$ is uniformly continuous, it suffices to identify $\Phi(g,t)$ for every $g$ and every $t$ in the open dense set $\mathcal{T}_0\equiv\{t:t_s\ne0\text{ for all }s\}$.
Fix such $(g,t)$ for the remainder of the proof, and define the scalar variables $V_{is}(t_s)\equiv t_s'\beta_{is}$, $s=1,\ldots,T$.
By construction, $\Phi(g,t)$ is the c.f. of $(\,\gamma_i,V_{i1}(t_1),\ldots,V_{iT}(t_T)\,)$ evaluated at $(g,1,\ldots,1)$, so it
suffices to identify the joint density of $(\gamma_i,V_{i1}(t_1),\ldots,V_{iT}(t_T))$.\footnote{
    This density exists under \Cref{ass:B1}\eqref{ass:B1:i}, because $V(t)$ is a linear function of $(\gamma_i,\beta_{i1}',\ldots,\beta_{iT}')$.}

\noindent \textbf{(\textit{Step 1.})}
For the fixed vector $t\in\mathbb R^{TJ_x}$, and for any $\lambda\in\mathbb{R}$, we can construct a directional vector \[\tilde t(\lambda) \equiv (t_1',t_2',\ldots,t_{T-1}',\lambda t_T')',\] and evaluate the vector $(S_{i1}(d),\ldots,S_{iT}(d))$ in the direction $\tilde t(\lambda)$.
Recall that $S_{is}(d)=p_s(\gamma_i;d)\beta_{is}$ is a scalar multiple of $\beta_{is}$, and $p_T(\gamma;d)=d_T$.
Therefore,
\[ \tilde t_s(\lambda)'S_{is}(d) =p_s(\gamma_i;d)\;\tilde t_s(\lambda)'\beta_{is} =
\begin{cases}
p_s(\gamma_i;d)\,V_{is}(t_s), & s=1,\ldots,T-1,\\[2pt]
d_T\,\lambda V_{iT}(t_T), & s=T.
\end{cases} \]
To simplify notation, we suppress $t_s$ in $V_{is}(t_s)$ for $s=1,\ldots,T$ below.

Summing over $s$ and substituting
$p_s(\gamma;d)=\sum_{r=s}^Td_r\gamma^{r-s}$,
\begin{align*}
\sum_{s=1}^T \tilde t_s(\lambda)'S_{is}(d)
&=\sum_{s=1}^{T-1}\ \sum_{r=s}^{T}d_r\,\gamma_i^{\,r-s}V_{is}
\;+\;d_T\,\lambda V_{iT}
=\sum_{r=1}^{T}d_r\ \sum_{s=1}^{\min(r,\,T-1)}\gamma_i^{\,r-s}V_{is}
\;+\;d_T\,\lambda V_{iT}\\[2pt]
&=\sum_{r=1}^{T-1}d_r
\underbrace{\sum_{s=1}^{r}\gamma_i^{\,r-s}V_{is}}_{\equiv \;W_r}
\;+\;d_T
\underbrace{\sum_{s=1}^{T-1}\gamma_i^{\,T-s}V_{is}}_{\equiv \;\gamma_iW_{T-1}}
\;+\;d_T\,\lambda V_{iT},
\end{align*}
where the second equality interchanges the order of summation over
$\{(s,r):1\le s\le T-1,\ s\le r\le T\}
=\{(s,r):1\le r\le T,\ 1\le s\le\min(r,T-1)\}$, and the third separates the term $r=T$, whose inner sum is
$\sum_{s=1}^{T-1}\gamma_i^{\,T-s}V_{is} =\gamma_i\left(\sum_{s=1}^{T-1}\gamma_i^{\,T-1-s}V_{is}\right)\equiv\gamma_iW_{T-1}$.
Note $W_r(\cdot) $ is indexed by $(t_1,\ldots,t_r)$ through $(V_{i1},\ldots,V_{ir})$, which we suppress in notation.

Collecting terms,
\begin{equation}\label{eq:family}
    \sum\nolimits_{s=1}^T \tilde t_s(\lambda)'S_{is}(d) = \sum\nolimits_{r=1}^{T-1}d_r\,W_r + d_T\bigl(\gamma_iW_{T-1}+\lambda V_{iT}\bigr), \qquad d\in\mathbb{R}^T.
\end{equation}
As noted above, the joint distribution of $(S_{i1}(d),\ldots,S_{iT}(d))$ given $Y_{i0}$ is identified for all $d\in\mathbb R^T$ under \Cref{ass:B1}. Therefore, for any fixed $t\in\mathbb R^{TJ_x}$, the distribution of the left-hand side of \eqref{eq:family} is identified for every $(d,\lambda)\in\mathbb R^{T+1}$. (Recall that $\tilde t(\lambda)$ is constructed from the direction vector $t$ in the c.f. of interest, $\Phi(g,t)$.)

Hence, the quantity
\begin{equation} \label{eq:cf-Xi}
    E\Bigl[\exp\Bigl\{i\Bigl(\sum_{r=1}^{T-1} d_r W_r + d_T \gamma_i W_{T-1} + d_T \lambda V_{iT} \Bigr)\Bigr\}\Bigr] = \phi_{\,\Xi}(d_1,\ldots,d_{T-1},\,d_T,\,d_T\lambda)
\end{equation}
is identified for all $(d,\lambda)\in\mathbb{R}^{T+1}$, where $\phi_{\,\Xi}(\cdot)$ is the characteristic function of
\[ \Xi \equiv (W_1,W_2,\ldots,W_{T-1},\gamma_iW_{T-1},V_{iT})\in\mathbb R^{T+1}.\]
Note $\Xi$ is a random array indexed by the fixed vector $t\in\mathbb R^{TJ_x}$ alone, and \emph{not} by the vector $(d,\lambda)$; the latter only determines the direction argument entering $\phi_{\,\Xi}(d_1,\ldots,d_T,\lambda d_T)$.

By varying $d$ over the set $\{d\in\mathbb R^T: d_T \neq 0 \}$ and $\lambda$ over $\mathbb R$, the last argument $z\equiv d_T\lambda$ in \eqref{eq:cf-Xi} ranges over $\mathbb R$. So, we can use \eqref{eq:cf-Xi} to recover the c.f. of $\Xi$ over the set $\{ (d_1,\ldots,d_T,z)\in\mathbb R^{T+1}:d_T\neq 0 \} $, which is dense in $\mathbb R^{T+1}$.
Characteristic functions are uniformly continuous, so $\phi_{\,\Xi}$ is identified on $\mathbb{R}^{T+1}$.
By L\'evy's uniqueness theorem, this determines the distribution of $\Xi$.

\noindent \textbf{(\textit{Step 2.})}
Define $\psi(\gamma,v_1,\ldots,v_T)\equiv (w_1,\ldots,w_{T-1},\,\gamma w_{T-1},\,v_T)$ with $w_r\equiv\sum_{s=1}^r\gamma^{r-s}v_s$, so that $\psi(\gamma_i,V_{i1},\ldots,V_{iT})=\Xi$.
Over $\{w_{T-1}\ne0\}$, $\psi$ is injective: $\gamma$ is recovered as the ratio of the $T$-th to the $(T{-}1)$-th coordinate, and then $v_1=w_1$, $v_r=w_r-\gamma w_{r-1}$ for $r=2,\ldots,T-1$, while $v_T$ is the last coordinate.
Decompose the mapping $\psi$ as the composite of two: $\psi_2\circ\psi_1$, where $\psi_1(\gamma,v_1,\ldots,v_T)\equiv(\gamma,w_1,\ldots,w_{T-1},v_T)$ is triangular with unit Jacobian, and $\psi_2(\gamma,w_1,\ldots,w_{T-1},v_T)\equiv(w_1,\ldots,w_{T-1},\gamma w_{T-1},v_T)$ has Jacobian determinant $(-1)^{T+1}w_{T-1}$. Hence the Jacobian of $\psi$ satisfies $|\!\det \mathrm D\psi|=|w_{T-1}|$.
Under \Cref{ass:B1}\eqref{ass:B1:i}, $\Pr\{W_{T-1}=0\}=0$.
Therefore, for almost every
$(\gamma,v)\in\mathbb{R}^{T+1}$,
\begin{equation*}
    f_{\gamma_i,V_{i1},\ldots,V_{iT}}(\gamma,v_1,\ldots,v_T) = f_\Xi\bigl(w_1,\ldots,w_{T-1},\gamma w_{T-1},v_T\bigr) \,\bigl|w_{T-1}\bigr|,
\end{equation*}
where $ w_r=\sum_{s=1}^r\gamma^{r-s}v_s $. This identifies the joint distribution of $(\gamma_i,V_{i1},\ldots,V_{iT})$.

\noindent \textbf{(\textit{Step 3.})} By construction, for any fixed vector $(g,t)$,
    \[ \Phi(g,t) = \int_{\mathbb{R}^{T+1}} \exp\Bigl\{i\Bigl(g\gamma+\sum\nolimits_{s=1}^Tv_s\Bigr)\Bigr\} \,f_{\gamma_i,V_{i1},\ldots,V_{iT}}(\gamma,v)\,d\gamma\,dv, \]
where $t$ is the fixed, directional vector defining $(V_{i1},\ldots,V_{iT})$.
(Recall that such dependence is suppressed in notation.)
With the distribution of $ (\gamma_i,V_{i1},\ldots,V_{iT}) $ recovered in Step 2, $\Phi(g,t)$ is identified for every $g\in\mathbb{R}$ and $t\in\mathcal{T}_0$, and by continuity on all of $\mathbb{R}^{1+TJ_x}$.
Since $\Phi$ is the joint characteristic function of $(\gamma_i,\beta_{i1},\ldots,\beta_{iT})$, the conclusion follows.
\end{proof}

\vspace{.3cm}

\section{Relation to Existing Literature} \label{sec:Appendix_lit}

\citet{GrahamPowell2012} study a correlated random coefficient panel data model:
\[Y_t = \chi_t'\mathcal{B}_t, \ \text{where }E(\mathcal B_t\mid \chi^T) = \mathcal{B}_0(\chi^T) + \mathcal D_t, \]
where $\mathcal B_0(\cdot) $ is a deterministic function that does not vary over time and $\mathcal D_t$ is a vector of constants that may vary over $t=1,2,\ldots,T$.
Mapping their specification into the reduced form of our model with $T=2$, we have $\chi_1 \equiv (1,Y_0,X_1,W_1)'$, $\chi_2 \equiv (1,Y_0,X_1,X_2,W_1,W_2)'$ and $\mathcal B_1 \equiv (U_1,\gamma,\beta_1,\delta_1)'$, $\mathcal B_2 \equiv (\gamma U_1 + U_2, \gamma^2, \gamma\beta_1, \beta_2, \gamma\delta_1, \delta_2)'$.

Our model is not nested within \citet{GrahamPowell2012}, because the mean of the random intercept in $\mathcal B_2$, i.e., $E(\gamma U_1+U_2 \mid \chi^T)$, is not additive in any stationary function $\mathcal B_0(\chi^T)$ and a constant $\mathcal D_2$ due to the sequential exogeneity of $W_t$ in $\chi_t$.
    Besides, \citet{GrahamPowell2012} study the (ir)regular identification of average partial effects $E(\mathcal B_t)$ when the number of time periods is no smaller than the dimension of covariates, using cross-section units whose covariates change little over time.
In contrast, we assume the coefficients for sequentially exogenous covariates are constant, and are interested in recovering the joint distribution of the random coefficients and intercepts. Our method does not require the number of time periods to be strictly larger than the dimension of exogenous covariates.

\citet{MastenTorgovitsky2016} study the identification of correlated random coefficient models in the following form:
\begin{equation} \label{eq:MT-restat2016}
    Y^*=B^*_0+B^*_1X^*+B^*_2Z_2^*, \qquad X^*=h(Z^*,V^*),
\end{equation}
where $h(\cdot)$ is an unknown function with a \textit{scalar} noise $V^*$, and $Z^*\equiv(Z^*_1,Z^*_2)$ is a vector of instruments.
The identifying conditions include:
\begin{align*}
         (a)  \ Z^* \perp (V^*,B^*)\text{, where } B^*\equiv(B^*_0,B^*_1,B^*_2); \
     (b)  \  h(Z^*,V^*) \text{ is increasing in }V^*\in\mathbb R.
\end{align*}
For simplicity, let $Y^*,X^*,Z^*_2$ all be scalars, and suppress the individual subscript $i$ below.
To fit our dynamic panel data model with random coefficients into the model of \citet{MastenTorgovitsky2016} in (\ref{eq:MT-restat2016}), one may relate the variables in our model to those in \citet{MastenTorgovitsky2016} as follows:
\[Y^*\equiv Y_2, \ X^*\equiv Y_1, \ Z^*_2\equiv X_2, \ Z^*_1 \equiv X_1=; \ B^*_1=\gamma, \ B^*_2 \equiv \beta_2, \ B^*_0\equiv U_{2}.\]
To simplify comparison, let the initial condition $Y_0$ in our model be degenerate at zero so that the structural equation in the first period is simplified to $Y_1 = X_1\beta_1 + U_1$, where $\beta_1$ is a random coefficient.
However, it is known from \citet{Imbens2007} and \citet{Kasy2011} that this cannot be fitted into the control function framework above.

\citet{HoderleinHolzmannMeister2017} consider a triangular model with random coefficients:
\begin{align*}
        \widetilde Y &= \widetilde B_0 + \widetilde B_1 \widetilde X + \widetilde B_2'\widetilde W, \\
    \widetilde X &= \widetilde A_0 + \widetilde A_1'\widetilde Z + \widetilde A_2'\widetilde W,
\end{align*}
where $\widetilde X, \widetilde Y$ are scalars. Let $\widetilde A\equiv (\widetilde A_0, \widetilde A_1', \widetilde A_2')'$; likewise for $\widetilde B$. The main identifying conditions in \citet{HoderleinHolzmannMeister2017} (Theorem 7) include:
\begin{align*}
(i)  \ (\widetilde Z',\widetilde W')'\perp(\widetilde A',\widetilde B')'; \
     (ii) \ \text{the support of} \ \widetilde Z \text{ contains an open interval}; \
     (iii)  \ \widetilde A_1 \perp \widetilde B.
\end{align*}
To map this setup into our model (with $T=2$ and suppressing $i$), let
\begin{align*}
    &  \widetilde Y \equiv Y_2 , \ \widetilde X\equiv Y_1 , \ \widetilde Z \equiv Y_0 , \ \widetilde W \equiv X \equiv (X_{1},X_{2}),\\
    &  \widetilde B_1\equiv\gamma , \ \widetilde B_2\equiv (0,\beta_2)', \ \widetilde B_0 \equiv U_2, \ \widetilde A_1 \equiv \gamma, \ \widetilde A_2 \equiv (\beta_1,0)', \ \widetilde A_0 \equiv U_1. \nonumber
\end{align*}
 Our identification result differs from \citet{HoderleinHolzmannMeister2017} in that we relax (i)--(iii) within the specific context of a random-coefficient dynamic panel data model, and also identify the joint distribution of random \textit{intercepts}.
First, we allow the random coefficients $(\widetilde A',\widetilde B')'$, which consist of $(\gamma,\beta,U_1,U_2)$ in our case, to be correlated with $\widetilde Z$, which is $Y_0$ in our case.
Second, we allow $\widetilde Z$ ($Y_0$) to have limited or even discrete support. In fact, our identification strategy conditions on $\widetilde Z$ ($Y_0$) and only exploits strict exogeneity in $\widetilde W$ ($X_i$).
Third, in our model, $\widetilde A_1=\widetilde B_1=\gamma$, which is allowed to be correlated with $(\beta,U_1,U_2)$ in $(\widetilde A_0, \widetilde B_0, \widetilde A_2', \widetilde B_2')$.
Last but not least, we also identify the \textit{joint} distribution of $(\widetilde A_0,\widetilde B_0)=(U_1,U_2)$ and $(\gamma,\beta)$ in $(\widetilde A_1, \widetilde A_2', \widetilde B_1, \widetilde B_2')$, whereas Theorem 7 in \citet{HoderleinHolzmannMeister2017} identifies the joint distribution of $(\widetilde B_0,\widetilde B_1)$.\footnote{Theorem 5 in \citet{HoderleinHolzmannMeister2017} identifies the joint distribution of $(\widetilde B_0, \widetilde B_1, \widetilde B_2')'$ under a stronger condition that the support of $(\widetilde Z',\widetilde W')'$ is $\mathbb R^{\dim(\widetilde Z)+\dim(\widetilde W)}$.}

To be clear, we are able to obtain these results without (i)--(iii) because of the specific structure in the random-coefficient dynamic panel data model.

\citet[Proposition 4]{masten2018random} studied a triangular system
\begin{align*}
    Y_2 & =   \gamma  Y_1  +  \delta_2^{*'} W^* + U_2,\\
    Y_1 & =   \beta_1 X_1  + \delta_1^{*'} W^* + U_1,
\end{align*}
where \(X_1 \perp (\gamma,\beta_1,\delta_1^{*},\delta_2^{*},U_1,U_2)\mid  W^*\), and identified the distribution of $(\gamma,\beta_1)\mid W^*$.

Our model in (\ref{eq:seq-exo-w}) can be partly reconciled with that of \citet{masten2018random}, with $W^*$ consisting of $(W_{i1},W_{i2},Y_{i0})'$, $\delta^*_1\equiv (\delta_1,0,\gamma)'$, $\delta^*_2 \equiv (0,\delta_2,0)'$, and the random coefficient for $Y_0$ in the structural form of $Y_1$ identical to that for $Y_1$ in the structural form of $Y_2$.
Nonetheless, our model is non-nested with \citet{masten2018random}, because we include a strictly exogenous regressor $X_2$ with a random coefficient $\beta_2$ in the equation for $Y_2$ and identify the joint distribution of $(\gamma,\beta_1,\beta_2)$.

Yet an even more important distinction between our work and \citet{masten2018random} is that we are able to identify the joint distribution of these random coefficients \textit{and the random intercepts} $(U_1,U_2)$, under the additional assumptions that $\delta_1,\delta_2$ are constants,\footnote{
    \citet{masten2018random} does not provide results for identifying the distribution of these random coefficients $(\delta_1,\delta_2)$ in his model.}
and that the random coefficients and intercepts are also conditionally independent from $W_{i1}$.

\bibliographystyle{apecon}
\bibliography{ref}

\end{document}